\documentclass[letterpaper]{article} 
\usepackage{aaai2026}  
\usepackage{times}  
\usepackage{helvet}  
\usepackage{courier}  
\usepackage[hyphens]{url}  
\usepackage{graphicx} 
\usepackage{natbib}  
\usepackage{caption} 
\usepackage{textcomp}
\usepackage{booktabs}
\usepackage{amsmath}
\usepackage{amsfonts}
\usepackage{amsthm}
\usepackage{algorithm}
\usepackage{subcaption}
\usepackage{algpseudocode}
\usepackage{csquotes}
\usepackage{enumitem}
\usepackage{multirow}
\usepackage{multicol}
\usepackage{tabularx}
\usepackage{array}
\usepackage{comment}
\usepackage{mathrsfs}
\usepackage{gensymb}
\usepackage{bbold}
\usepackage{mathtools}
\newtheorem{definition}{Definition}
\newtheorem{remark}{Remark}
\newtheorem{corollary}{Corollary}

\newtheorem{theorem}{Theorem}

\usepackage{cite}
\usepackage{amssymb}
\usepackage{bm}            

\def \eg{\emph{e.g.}}
\def \ie{\emph{i.e.}}

\newcommand{\R}{\mathbb R}            
\DeclareMathOperator{\diag}{diag}

\def \1{\mathbb{1}}

\def \grr{\textsf{RandomizedResponse}} 
\def \mss{\textsf{ModularSubsetSelection}} 
\def \pgr{\textsf{ProjectiveGeometryResponse}}

\def \rappor{\textsf{RAPPOR}}
\def \oue{\textsf{OptimalUnaryEncoding}}

\def \ss{\textsf{SubsetSelection}}
\def \rappor{\textsf{RAPPOR}}

\usepackage{newfloat}
\usepackage{listings}
\DeclareCaptionStyle{ruled}{labelfont=normalfont,labelsep=colon,strut=off} 
\floatstyle{ruled}
\newfloat{listing}{tb}{lst}{}
\floatname{listing}{Listing}
\title{Private Frequency Estimation Via Residue Number Systems}
\author {
    Héber H. Arcolezi
}
\affiliations {
    Inria, Grenoble, France\\
    heber.hwang-arcolezi@inria.fr
}

\begin{document}
\sloppy

\maketitle

\begin{abstract}
We present \textsf{ModularSubsetSelection} (MSS), a new algorithm for locally differentially private (LDP) frequency estimation. 
Given a universe of size $k$ and $n$ users, our $\varepsilon$-LDP mechanism encodes each input via a Residue Number System (RNS) over $\ell$ pairwise-coprime moduli $m_0, \ldots, m_{\ell-1}$, and reports a randomly chosen index $j \in [\ell]$ along with the perturbed residue using the statistically optimal \textsf{SubsetSelection} (SS)~\cite{wang2016mutual}.
This design reduces the user communication cost from $\Theta\bigl(\omega \log_2(k/\omega)\bigr)$ bits required by standard SS (with $\omega \approx k/(e^\varepsilon+1)$) down to $\lceil \log_2 \ell \rceil + \lceil \log_2 m_j \rceil$ bits, where $m_j < k$.
Server-side decoding runs in $\Theta(n + r k \ell)$ time, where $r$ is the number of LSMR~\cite{fong2011lsmr} iterations. 
In practice, with well-conditioned moduli (\textit{i.e.}, constant $r$ and $\ell = \Theta(\log k)$), this becomes $\Theta(n + k \log k)$. 
We prove that MSS achieves worst-case MSE within a constant factor of state-of-the-art protocols such as SS and \textsf{ProjectiveGeometryResponse} (PGR)~\cite{feldman2022}, while avoiding the algebraic prerequisites and dynamic-programming decoder required by PGR.
Empirically, MSS matches the estimation accuracy of SS, PGR, and \textsf{RAPPOR}~\cite{rappor} across realistic $(k, \varepsilon)$ settings, while offering faster decoding than PGR and shorter user messages than SS.
Lastly, by sampling from multiple moduli and reporting only a single perturbed residue, MSS achieves the lowest reconstruction-attack success rate among all evaluated LDP protocols. 
\end{abstract}

\begin{links}
    \link{Code}{https://github.com/hharcolezi/private-frequency-oracle-rns}
\end{links}

\section{Introduction}
\label{sec:intro}

Today’s \emph{federated} applications span billions of devices, such as keyboard prediction by Apple~\cite{apple} and Gboard~\cite{gboard}, and telemetry systems in Google Chrome~\cite{rappor} and Microsoft operating systems~\cite{microsoft}, all of which must learn from data that never leaves the user’s device in raw form.
The prevailing formalism is the \emph{local model} of differential privacy (LDP)~\cite{first_ldp,Duchi2013}: each user applies a randomizer \(\mathcal M\) to their datum \(x\in\mathcal X\) and sends only the noisy message \(Y=\mathcal M(x)\) to an \emph{untrusted} aggregator.  
A mechanism \(\mathcal M\!:\!\mathcal X\!\to\!\mathcal Y\) is \(\varepsilon\)\,-\,LDP if for every measurable \(S\subseteq\mathcal Y\) and every \(x,x'\in\mathcal X\),
\[
  \Pr[\mathcal M(x)\in S]\;\le\;e^{\varepsilon}\Pr[\mathcal M(x')\in S].
\]
Under LDP \emph{no single report} can distinguish two inputs by a factor larger than \(e^{\varepsilon}\).
Once the locally obfuscated reports arrive, the server aims to perform global tasks such as statistical estimation or model training.  
Four main factors determine the practicality of any local-DP protocol:

\begin{enumerate}[label=(\roman*)]
  \item \textbf{Utility}: the accuracy with which the server can complete its task.
  \item \textbf{Communication}: the number of bits each user must transmit per report.
  \item \textbf{Server runtime}: the time and memory required for server-side decoding.
  \item \textbf{Attackability}: the probability that an adversary correctly recovers an individual’s input from a single report.
\end{enumerate}

Together, these four dimensions form a \emph{multi-constraint regime}: in large-scale telemetry and federated analytics, client bandwidth, server compute, statistical accuracy, and privacy risk may each become the dominant constraint depending on the deployment.

\paragraph{Problem Statement.} This work addresses this \emph{multi-constraint} challenge when designing efficient mechanisms for \emph{federated-analytics} deployments that require locally differentially private \textit{\textbf{frequency estimation}} over a finite domain $[k] = \{0, \dots, k-1\}$.  
In this setting, each user holds a private input $x_i$ and the goal is for an untrusted server to recover an accurate estimate of the population histogram $\mathbf{f} \in \R^k$, where $f_v = \tfrac{\#\{i : x_i = v\}}{n}$.  
After collecting $n$ randomized reports $\{Y_i\}_{i=1}^n$, the server computes an estimate $\hat{\mathbf{f}}$ aiming to minimize its distance from $\mathbf{f}$ under some norm $\lVert \mathbf{f} - \hat{\mathbf{f}} \rVert$. 
In line with prior literature~\cite{feldman2022,kairouz2016discrete,tianhao2017,Hadamard}, we quantify estimation error using the expected $\ell_2$ norm, and focus on the mean squared error (MSE) metric: \(\mathrm{MSE}\;=\;\frac1k\,\operatorname E\bigl[\lVert\hat{\mathbf f}-\mathbf f\rVert_2^{2}\bigr]\).

In addition to utility, another fundamental concern in the local DP model is \emph{\textbf{attackability}}: the ability of a Bayesian adversary to reconstruct a user’s true input $x$ from a single obfuscated message $Y$~\cite{Gursoy2022,arcolezi2025revisiting}.  
This threat, commonly referred to as a \emph{Data Reconstruction Attack (DRA)} in the AI and ML communities~\cite{geiping2020inverting,hayes2023bounding,Guerra2024}, is quantified as the probability that an adversary with full knowledge of the protocol and prior distribution correctly guesses $x$ given $Y$. 
Protocols that minimize estimation error while keeping reconstruction rate low provide stronger privacy in practice.

\begin{table*}[t]
  \centering
  \resizebox{\linewidth}{!}{%
  \begin{tabular}{@{}lcccc@{}}
    \toprule
    \textbf{LDP frequency-oracle} &
      \textbf{Communication (bits)} &
      \textbf{MSE (worst-case)} &
      \textbf{Server decoding time} &
      \textbf{Attackability (DRA)} \\[8pt]
    \midrule
    \grr~(GRR)~\cite{kairouz2016discrete}
      & $\lceil\log_2 k\rceil$
      & $\dfrac{e^{\varepsilon}+k-2}{n\,(e^{\varepsilon}-1)^2}$
      & $O(n+k)$
      & $\dfrac{e^{\varepsilon}}{e^{\varepsilon}+k-1}$ \\[8pt]

    \rappor{}~\cite{rappor,tianhao2017}
      & $k$
      & $\dfrac{4\,e^{\varepsilon}}{n\,(e^{\varepsilon}-1)^2}$
      & $O(nk)$
      & $\frac{1}{k}\left[
        e^{\varepsilon/2}
        -
        \frac{e^{(k-1)\varepsilon/2}\bigl(e^{\varepsilon/2}-1\bigr)}
             {\bigl(e^{\varepsilon/2}+1\bigr)^{k-1}}
        \right]$ \\ [8pt]

    \ss~(SS)~\cite{wang2016mutual}
      & $\Bigl\lceil\log_{2}\binom{k}{\omega}\Bigr\rceil$
      & $\dfrac{4\,e^{\varepsilon}}{n\,(e^{\varepsilon}-1)^2}$
      & $O(n\omega+k)$
      & $\dfrac{e^{\varepsilon}}{\omega e^{\varepsilon}+k-\omega}$ \\ [8pt]

    \pgr~(PGR)~\cite{feldman2022}
      & $\lceil\log_2 k\rceil$
      & $\dfrac{4\,e^{\varepsilon}}{n\,(e^{\varepsilon}-1)^2}$
      & $O\!\bigl(n+k\,e^{\varepsilon}\log k\bigr)$
      & $\displaystyle \frac{e^{\varepsilon}}{K + (e^{\varepsilon} - 1)c_{\mathrm{set}}}$ \\[8pt]

    \textbf{\mss} (this work)
      & $\displaystyle
         \lceil\log_2\ell\rceil+
         \frac{1}{\ell}\sum_{j=0}^{\ell-1}\!
         \Bigl\lceil\log_{2}\binom{m_j}{\;\omega_j}\Bigr\rceil$
      & $\displaystyle
         \frac{4\,\kappa\,e^{\varepsilon}}
              {n\,(e^{\varepsilon}-1)^2}$
      & $O\!\bigl(n + k\ell + \sum_{j=0}^{\ell-1} m_j\bigr)$
      & $\displaystyle
            \frac{1}{\ell k} \sum_{j=0}^{\ell - 1}
            \frac{m_j \cdot e^{\varepsilon}}{\omega_j \cdot e^{\varepsilon} + m_j - \omega_j}$\\[8pt]
    \bottomrule
  \end{tabular}}
  \caption{Comparison of single-message LDP frequency-estimation schemes.
        Communication is the number of bits per user.
        MSE is the worst-case mean-squared error of the unbiased estimator; server time is leading-order in users $n$ and domain size $k$.
        DRA is the Bayesian single-message attacker success rate.
        $\omega=\lfloor k/(e^{\varepsilon}+1)\rceil$ is the SS subset size, and $\omega_j$ is the analogous size for modulus $m_j$ in MSS.
        For PGR, the DRA expression shown applies when the domain equals the natural projective size $K=(q^{t}-1)/(q-1)$; the exact DRA for truncated domains ($k < K$) is provided in 
        Appendix~\ref{app:attack_pgr}.
        }
  \label{tab:ldp_summary}
\end{table*}

\paragraph{Related work.} Table~\ref{tab:ldp_summary} summarizes the trade-offs across utility, communication, computation, and attackability of state-of-the-art LDP frequency estimation protocols.
Classical \grr~\cite{Warner1965,kairouz2016discrete}  minimizes
per-user bandwidth (one $\lceil\log_{2}k\rceil$-bit symbol) but suffers an
$\Theta(k/e^{\varepsilon})$ gap to the information-theoretic MSE bound and
yields the highest single-message reconstruction success rate.
The \ss~(SS) mechanism~\cite{wang2016mutual} attains the optimal worst-case MSE by returning a random subset containing the true value. 
However, this comes with $\Theta\!\bigl(\omega\log_{2}(k/\omega)\bigr)$ bits of communication per user and high server cost.  
Bit-vector schemes like \rappor~\cite{rappor} and \oue~(OUE)~\cite{tianhao2017} reach the same optimal bound by perturbing $k$-length binary encodings, but this increases both message size ($O(k)$) and server time ($O(nk)$).  
Most recently, the coding-based \pgr{} protocol~\cite{feldman2022} demonstrates that algebraic structure can enable near-optimal utility with reduced communication cost as $\lceil \log_2 k \rceil$. 
However, its deployment remains nontrivial: PGR requires the domain size to match a projective geometry constraint, relies on finite fields of size near $e^{\varepsilon}$, and uses dynamic programming for decoding.

\paragraph{Our Contributions.} 
We propose \mss{} (MSS), a novel single-message $\varepsilon$-LDP protocol that tackles the accuracy-bandwidth-computation-attackability four-way trade-off through a modular \emph{``divide \& conquer''} design based on \emph{Residue Number System} (RNS)~\cite{Szabó1967}. 
Each input \(x \in [k]\) is first mapped to a short RNS vector \(\bigl(x \bmod m_0, \dots, x \bmod m_{\ell-1}\bigr)\) using a set of pairwise-coprime moduli \((m_0,\dots,m_{\ell-1})\); by the Chinese Remainder Theorem (CRT), 
\(\prod_{j=0}^{\ell-1} m_j \ge k\) ensures that the mapping is injective over $[k]$.  
Instead of transmitting the full residue vector, each user samples \emph{one} block index uniformly at random and perturbs its coordinate with \ss{} at privacy level $\varepsilon$.  
This ``divide'' step reduces the message alphabet from $k$ to at most $\max_j m_j < k$, so the report fits into \(\lceil\log_2 \ell\rceil + \lceil\log_2 m_j\rceil\) bits. 
\emph{On the user side, this requires nontrivial CRT-based design choices to maintain 
injectivity, full rank, and a favorable $\ell$-$m_j$ trade-off.}

On the server side, MSS ``conquers'' the estimation error via a variance-weighted least-squares solver on a sparse design matrix $A_{w}$.  
For well-conditioned moduli, the total decoding cost is $O(n+k\ell)$ (empirically $O(n+k\log k)$).
Theoretically, the worst-case mean-squared error satisfies 
\(
\mathrm{MSE}_{\text{MSS}}
\;\le\;
\kappa\,
\mathrm{MSE}_{\text{SS}},
\quad \kappa=\operatorname{cond}(A_{w}),
\)
and our modulus search keeps $\kappa \leq 10$.  
\emph{The server-side challenges include controlling $\kappa$ to guarantee low MSE, 
designing a variance-optimal unbiased decoder, and selecting moduli that balance accuracy and computational cost.}
In practice, the \emph{empirical} ratio $\mathrm{MSE}_{\text{MSS}} / \mathrm{MSE}_{\text{SS}}$ never exceeded $\kappa \approx 1.3$ across all $(k, \varepsilon)$ we tested, indicating only a small constant-factor overhead.
Lastly, MSS reduces a Bayesian attacker’s single-report reconstruction success by increasing uncertainty over the domain, outperforming \ss{} and \grr{} in our experiments.

\paragraph{Comparison with \pgr.}
PGR~\cite{feldman2022} attains the information-theoretic variance bound of SS but at the cost of finite-field arithmetic, rigid domain constraints, and a dynamic-programming decoder with $O\!\bigl(n+k\,e^{\varepsilon}\log k\bigr)$ states.  
MSS eliminates these algebraic prerequisites: it accepts \emph{arbitrary} $k$ and $\varepsilon$, relies solely on integer arithmetic, and replaces combinatorial decoding by a single sparse least-squares solve amenable to out-of-core and parallel settings.  
Empirically, MSS matches or approximates the utility loss of PGR (and SS) while requiring much less server-side runtime; moreover, the tunable parameter $\ell$ lets practitioners navigate the full communication-accuracy spectrum, a flexibility unavailable in PGR.  
Therefore, MSS offers a lower-complexity, more adaptable alternative without sacrificing practical accuracy.

\section{Preliminaries}
\label{sec:background}

Our MSS combines ideas from number theory with tools from linear algebra. 
We review the background below.

\begin{definition}[Residue Number System (RNS)~\cite{Szabó1967}] \label{def:rns}
Let $\mathcal{X} = \{0, \dots, k-1\}$ be the finite input domain.
Given a set of pairwise-coprime integers $\{m_0, \dots, m_{\ell-1}\}$, called \emph{moduli}, the Residue Number System represents each $x \in \mathcal{X}$ by its residues modulo each $m_j$:
\[
   \mathbf{r}(x) = \bigl(x \bmod m_0,\; \dots,\; x \bmod m_{\ell-1}\bigr).
\]
By the Chinese Remainder Theorem, if $\prod_{j=0}^{\ell-1} m_j \ge k$, this representation is injective and fully encodes the domain.
\end{definition}

\paragraph{Weighted least-squares estimators.}
Consider noisy linear measurements \(y=Ax+\epsilon\), where \( A\in\R^{M\times k} \) is a design matrix and \(\epsilon\) has row-wise variances \((w_{1}^{-1},\dots,w_{M}^{-1})\).
The \emph{generalised least-squares} (GLS) estimator solves
\[
  \hat x
  \;=\;
  \arg\min_{z}\;\bigl\|W^{1/2}(Az-y)\bigr\|_{2}^{2}
  \;=\;
  (A^{\!\top}WA+\lambda I)^{-1}A^{\!\top}Wy,
\]
with weight matrix \(W=\operatorname{diag}(w_{1},\dots,w_{M})\) and optional ridge parameter \(\lambda\ge0\)~\cite{hastie2009elements}.

\paragraph{Spectral condition number.}
For any real matrix \(B\) let \(\sigma_{\max}(B)\) and \(\sigma_{\min}(B)\) denote its largest and smallest singular values. 
The \emph{spectral condition number} is
\[
  \kappa=\operatorname{cond}(B)
  \;=\;
  \frac{\sigma_{\max}(B)}{\sigma_{\min}(B)}
  \;\in\;[1,\infty).
\]
Smaller values imply greater numerical stability.
In our analysis, we write \(\kappa = \operatorname{cond}(A_w)\) for the weighted design matrix \(A_w = W^{1/2} A\).

\paragraph{Iterative solution.}
When \(A\) is large and sparse (for large domain sizes $k$), we solve the GLS normal equations with the Lanczos-based \textsc{LSMR} algorithm~\cite{fong2011lsmr}, whose cost is \(O\!\bigl(r\,\mathrm{nnz}(A)\bigr)\) where \(r\) is the iteration count to convergence and ``nnz'' is a shorthand for the number of non-zero entries in a matrix.

\section{Modular Subset Selection}
\label{sec:mss}

Following the \emph{divide \& conquer} view from Section~\ref{sec:intro},
Section~\ref{sub:mss_client} covers the user-side (\emph{divide}) mechanism,
and Section~\ref{sub:mss_estimation} the server-side (\emph{conquer}) estimation.

\subsection{User-Side (``Divide'') Obfuscation}
\label{sub:mss_client}

\mss{} is a single-message $\varepsilon$-LDP mechanism for frequency estimation 
over the domain $\mathcal{X}=\{0,\dots,k-1\}$. 
Building on the residue number system (Definition~\ref{def:rns}), each input $x$ is 
represented by its $\ell$ residues modulo a set of pairwise-coprime moduli. 
Rather than perturbing all residues with a split privacy budget, MSS samples and 
reports only a \emph{single} coordinate $J\in[\ell]$, using the full privacy budget 
$\varepsilon$ for that coordinate.  
This modular sampling aligns with established LDP approaches for multidimensional 
data~\cite{tianhao2017,Arcolezi2023} and yields strong privacy, low communication 
cost, and efficient server-side recovery.

Concretely, each user holding a private value $x$ proceeds by selecting one modulus $m_J$ 
uniformly at random, computing the residue $r = x \bmod m_J$, and applying 
\ss{} with privacy level~$\varepsilon$ over the domain $[m_J]$.  
The resulting report consists of the block index $J$ and a noisy subset 
$Z \subseteq [m_J]$ of fixed size $\omega_J$.  
The full procedure is given in Algorithm~\ref{alg:mss-client}.

\begin{theorem}[Privacy of MSS]
\mss{} in Algorithm~\ref{alg:mss-client} satisfies $\varepsilon$-local differential privacy.
\end{theorem}

\begin{proof}
Fix any $x,x' \in \mathcal{X}$ and any possible output $(j,Z)$.
The output of \mss{} consists of two components: (i) a uniformly sampled block index $J \in [\ell]$, and (ii) a perturbed residue set $Z \subseteq [m_j]$ of fixed size $\omega_j$ generated by \ss{}.

By construction,
\[
\Pr[\text{MSS}(x) = (j, Z)] = \Pr[J = j] \cdot \Pr[Z \mid J = j, x].
\]
Since $J$ is independent of $x$ and uniform over $[\ell]$, it contributes no privacy loss.
It suffices to show that for fixed $j$, the randomizer $\text{SS}_{m_j}$ applied to $x \bmod m_j$ satisfies $\varepsilon$-LDP.

Let $r = x \bmod m_j$ and $r' = x' \bmod m_j$. From the SS mechanism definition~\cite{wang2016mutual}, we have:
\[
\frac{
\Pr[Z \mid r]
}{
\Pr[Z \mid r']
}
\le e^{\varepsilon},
\quad
\forall Z \subseteq [m_j],\; |Z| = \omega_j.
\]
Hence,
\[
\frac{
\Pr[\text{MSS}(x) = (j, Z)]
}{
\Pr[\text{MSS}(x') = (j, Z)]
}
= \frac{1/\ell}{1/\ell} \cdot \frac{\Pr[Z \mid r]}{\Pr[Z \mid r']}
\le 1 \cdot e^{\varepsilon} = e^{\varepsilon}.
\]
Finally, since post-processing does not affect privacy, \mss{} satisfies $\varepsilon$-LDP.
\end{proof}

\begin{algorithm}[t]
\caption{$\textsc{UserSideMSS}(x, \, \mathbf{m}, \, \varepsilon)$}
\label{alg:mss-client}
\renewcommand{\baselinestretch}{1.0}\selectfont
\begin{algorithmic}[1]
\Require Private input $x \in \mathcal{X}$, moduli $\mathbf{m}$, privacy budget $\varepsilon$
\Ensure Noisy report $(J,Z)$
\State $\ell \gets |\mathbf{m}|$
\State Draw $J \sim \operatorname{Uniform}([\ell])$ \Comment{Sample modulus index}
\State Set $p_J \gets \frac{\omega_J e^{\varepsilon}}{\omega_J e^{\varepsilon} + m_J - \omega_J}$, where $\omega_j = \lfloor\frac{m_J}{e^{\varepsilon}+1}\rceil$
\State Compute $r \gets x \bmod m_J$
\State Draw $\zeta \sim \operatorname{Uniform}([0,1])$
\If{$\zeta < p_J$}
  \State $Z \gets \{r\} \cup$ random sample of $(\omega_J - 1)$ elements from $[m_J] \setminus \{r\}$
\Else
  \State $Z \gets$ random sample of $\omega_J$ elements from $[m_J] \setminus \{r\}$
\EndIf
\State \Return $(J,Z)$
\end{algorithmic}
\renewcommand{\baselinestretch}{1.0}\selectfont
\end{algorithm}

\subsection{Server-Side (``Conquer") Estimation}
\label{sub:mss_estimation}

Upon receiving the users' reports $y = (J, Z)$, the server's goal is to estimate the empirical distribution $\mathbf{f} \in \R^k$ over $[k]$. 
This is done by first debiasing the noisy SS reports, forming a weighted design matrix that leverages the CRT structure, and then solving a regularized least-squares system.

\subsubsection{Design Matrix}

For each block $j \in [\ell]$, define the mapping matrix $A_j \in \{0,1\}^{m_j \times k}$ such that
\[
  A_j[r, x] = \mathbf{1}\{x \bmod m_j = r\}.
\]
Each row of $A_j$ encodes the indicator vector of domain values that map to residue $r$ under modulus $m_j$.
Stacking all $A_j$ vertically produces the full design matrix:
\[
  A = \begin{bmatrix} A_0 \\ \vdots \\ A_{\ell-1} \end{bmatrix}
  \in \{0,1\}^{T \times k}, \quad T = \sum_{j=0}^{\ell-1} m_j.
\]

\subsubsection{Variance-Optimal Row Weights} 
To reflect the per-block variance from the \ss{} mechanism, we apply optimal variance weights to each row. 
Let $p_j=\frac{\omega_j e^{\varepsilon}}{\omega_j e^{\varepsilon} + m_j - \omega_j}$ and $q_j=\frac{\omega_j e^{\varepsilon} (\omega_j - 1) + (m_j - \omega_j) \omega_j}{(m_j - 1) (\omega_j e^{\varepsilon} + m_j - \omega_j)}$ denote the true and false inclusion probabilities for block $j$.
The marginal probability that a random residue appears in $Z$ is
\[
  \pi_j = q_j + \frac{p_j-q_j}{m_j},
\]
and for $n_j$ reports using block $j$, the variance of each SS estimator
coordinate is
\[
  \sigma_j^2=\frac{\pi_j(1-\pi_j)}{n_j(p_j-q_j)^2}.
\]

Following generalized least squares, we define the square-root weight vector $\mathbf{w}^{1/2} \in \R^T$ such that each entry corresponding to block $j$ is repeated $m_j$ times and equals
\[
  \sqrt{w_{j}} = \frac{p_j - q_j}{\sqrt{\pi_j(1 - \pi_j)/n_j}}.
\]

Let the diagonal scaling matrix:
\[
  W^{1/2} = \diag(\underbrace{\sqrt{w_0}, \dots, \sqrt{w_0}}_{m_0},
                 \dots,
                 \underbrace{\sqrt{w_{\ell-1}}, \dots, \sqrt{w_{\ell-1}}}_{m_{\ell-1}}).
\]
The weighted design matrix is then: 

\[A_w = W^{1/2} A \mathrm{.}\]

\subsubsection{Observation Vector Construction}

For each block $j$, let $c_j \in \R^{m_j}$ be the vector of counts, where $c_j[a]$ is the number of times residue $a \in [m_j]$ appeared in block $j$. 
Let $n_j = \sum_a c_j[a]$, define the empirical probability vector $\bar{y}_j = c_j / n_j$.
Following the unbiased SS estimator~\cite{wang2016mutual}, the debiased per-residue estimate is
\[
  \hat{s}_j = \frac{\bar{y}_j - q_j}{p_j - q_j},
\]
and each coordinate of $\hat{s}_j$ has variance $\sigma_j^2$ from above.

Stacking all blocks yields 
\[
\mathbf{s} = \begin{bmatrix} \hat{s}_0 \\[-2pt] \vdots \\[-2pt] \hat{s}_{\ell-1} \end{bmatrix}.
\]
Since each coordinate has variance $\sigma_j^2$, we apply the standard GLS
reweighting, \ie, scaling each entry by the inverse of its noise standard deviation, 
to obtain the weighted observations
\[
  \tilde{\mathbf{s}} = \mathbf{w}^{1/2} \odot \mathbf{s}.
\]

\subsubsection{Least Squares Estimation}

To estimate the raw frequency vector $\hat{\mathbf{f}} \in \R^k$, we solve:
\[
  \hat{\mathbf{f}} = \arg\min_{\mathbf{z} \in \R^k} \left\| A_w \mathbf{z} - \tilde{\mathbf{s}} \right\|_2^2 + \lambda \|\mathbf{z}\|_2^2,
\]
where $\lambda > 0$ is a small ridge regularization parameter (\eg, $1/\varepsilon^2$) introduced for numerical stability. 
In practice, we solve this system using the \textsc{LSMR} algorithm~\cite{fong2011lsmr}, a Krylov-subspace method that efficiently handles large sparse matrices and avoids explicit inversion. 
When $\lambda=0$ and $A_w$ has full column rank, the solution is:
\begin{equation} \label{eq:estimator}
  \hat{\mathbf{f}} = \left( A_w^\top A_w + \lambda I \right)^{-1} A_w^\top \tilde{\mathbf{s}},  
\end{equation}
but the computation is performed iteratively without forming dense matrices.
This estimator is unbiased in expectation (see Theorem~\ref{thm:unbiased-zero}) as well as asymptotically (see Corollary~\ref{cor:unbiased-lambda}), and its analytical variance is derived in Section~\ref{sub:mse}.

\subsubsection{Unbiasedness Analysis} \label{sub:unbiasedness}
An important property of any frequency oracle is whether its estimates are unbiased. 
For our MSS estimator, we now show that it satisfies exact unbiasedness in the unregularized case ($\lambda = 0$), and is asymptotically unbiased when a small ridge regularization is applied.

\begin{theorem}[Exact unbiasedness, $\lambda=0$]
\label{thm:unbiased-zero}
Let $\mathbf{f} \in \mathbb{R}^k$ be the true input histogram, and suppose the weighted design matrix $A_w \in \mathbb{R}^{T \times k}$ has full column rank.
Then the least-squares estimator
\[
  \hat{\mathbf{f}} = (A_w^\top A_w)^{-1} A_w^\top \tilde{\mathbf{s}}
\]
satisfies
\[
  \mathbb{E}[\hat{\mathbf{f}}] = \mathbf{f}.
\]
\end{theorem}

\begin{proof}
Recall that $\tilde{\mathbf{s}} = W^{1/2} \mathbf{s}$, and from the de-biasing procedure, each entry $s_{j,a}$ satisfies:
\[
  \mathbb{E}[s_{j,a}] = \sum_{x : x \bmod m_j = a} f_x = (A \mathbf{f})_{j,a}.
\]
Stacking over all blocks yields:
\[
  \mathbb{E}[\tilde{\mathbf{s}}] = W^{1/2} A \mathbf{f} = A_w \mathbf{f}.
\]
Taking the expectation of the estimator gives:
\[
  \mathbb{E}[\hat{\mathbf{f}}]
  = (A_w^\top A_w)^{-1} A_w^\top \mathbb{E}[\tilde{\mathbf{s}}]
  = (A_w^\top A_w)^{-1} A_w^\top A_w \mathbf{f}
  = \mathbf{f}.
\]
\end{proof}

\begin{corollary}[Asymptotic Unbiasedness, $\lambda > 0$]
\label{cor:unbiased-lambda}
Let $\hat{\mathbf{f}}_\lambda$ be the regularized estimator:
\[
  \hat{\mathbf{f}}_\lambda = (A_w^\top A_w + \lambda I)^{-1} A_w^\top \tilde{\mathbf{s}}.
\]
Then, as $\lambda \to 0$, the estimator converges in expectation to the true histogram:
\[
  \mathbb{E}[\hat{\mathbf{f}}_\lambda] \;\longrightarrow\; \mathbf{f}.
\]
For practical settings (\eg, $\lambda = 1/\varepsilon^2$), the bias introduced is $O(\lambda)$, which becomes negligible for large $\varepsilon$ (weaker privacy). 
When $\varepsilon$ is small (strong privacy), regularization bias may be more significant.
\end{corollary}

\begin{remark}
In practice, a small ridge term $\lambda > 0$ can be used to improve numerical conditioning and accelerate convergence of iterative solvers. 
Although this introduces a small bias, the estimator remains practically unbiased.
\end{remark}

\section{Analysis of MSS} \label{sec:analysis_mss}

This section analyzes MSS along the four axes of the multi-constraint regime highlighted in Section~\ref{sec:intro}.
We bound its communication and decoding costs, derive closed-form and worst-case MSE expressions in terms of the condition number $\kappa$, show how moduli selection controls $\kappa$, and quantify its resilience to data reconstruction attacks.

\subsection{Communication Cost} \label{sub:comm_cost}

Each user sends a pair $(J, Z)$ as defined in Algorithm~\ref{alg:mss-client}, where $J \in [\ell]$ is the block index and $Z \subseteq [m_J]$ is a noisy subset of size $\omega_J$ produced via SS.  
The number of bits is:

\begin{equation} \label{eq:mss_comm_cost}
    \left\lceil \log_2 \ell \right\rceil + \frac{1}{\ell} \sum_{j=0}^{\ell-1} \left\lceil \log_2 \binom{m_j}{\omega_j} \right\rceil.
\end{equation}

This reflects the average-case encoding cost, assuming uniform selection of the block index and optimal enumeration-based encoding over the $\binom{m_j}{\omega_j}$ possible subsets.

\subsection{Server-Side Decoding and Aggregation Cost}
\label{sub:mss_server_cost}

The total server-side runtime consists of three main phases: collecting residue counts, forming the debiased observation vector, and solving the weighted least-squares problem.

First, a single pass over the $n$ user reports is sufficient to aggregate per-block residue counts and compute normalization factors, requiring $O(n)$ time. Next, debiasing each empirical histogram $\bar{y}_j$ and applying variance-optimal weights to form the scaled vector $\tilde{\mathbf{s}}$ takes $O\bigl(\sum_{j=0}^{\ell-1} m_j\bigr)$ operations.

Finally, to recover the histogram $\hat{\mathbf{f}} \in \mathbb{R}^k$, the server solves a regularized least-squares system using the sparse weighted design matrix $A_w \in \mathbb{R}^{T \times k}$, where $T = \sum_{j} m_j$. 
As mentioned in Section~\ref{sub:mss_estimation}, we use \textsc{LSMR}~\cite{fong2011lsmr}, which takes $r$ iterations. 
Each iteration performs one matrix-vector multiplication with $A_w$ and $A_w^\top$, costing $O(k\ell)$ due to the structured sparsity of the CRT design.

\paragraph{Overall runtime.}
The total decoding complexity is thus:
\[
  O\left(n + k\ell + \sum_{j=0}^{\ell-1} m_j\right).
\]
In practice, one can select moduli such that $\sum_j m_j = O(k)$ (since each $m_j < k$ and $\ell = \Theta(\log k)$), yielding:
\[
  O(n + k\ell) \quad \text{with } \ell = \Theta(\log k).
\]

\paragraph{Simplified bounds.}
\begin{itemize}
  \item When the number of solver iterations is constant ($r = O(1)$), the runtime becomes: $\Theta(n + k \log k)$.
  \item In the worst case, when convergence requires $r = \Theta(k)$ iterations, the runtime becomes: $\Theta(n + k^2 \log k)$.
\end{itemize}

\subsection{Closed-form Variance of the Estimator}
\label{sub:mse}

Let $\mathbf f = (f_0, \dots, f_{k-1})^\top \in \mathbb{R}^k$ denote the unknown population histogram over domain $[k]$, satisfying $\sum_{x=0}^{k-1} f_x = 1$. 
For each modulus $m_j$ (with $j = 0, \dots, \ell - 1$), define the corresponding RNS marginal distribution:
\begin{equation*}
  g_{j,a} = \sum_{\smash{x:\,x \bmod m_j = a}} f_x,
  \qquad 
  \mathbf g_j = (g_{j,0}, \dots, g_{j, m_j-1})^\top.
\end{equation*}

\paragraph{Noisy subset selection.}

Conditioned on a user sampling block $J = j$, the probability that a specific residue $a \in [m_j]$ appears in the subset $Z$ is:
\begin{align*}
  \pi_{j,a} &= \mathbb{P}[a \in Z \mid J = j]
  = p_j \cdot g_{j,a} + q_j \cdot (1 - g_{j,a})\\
  &= q_j + (p_j - q_j) \cdot g_{j,a},
\end{align*}
where $p_j$ and $q_j$ are the true and false inclusion probabilities of \ss{} for block $j$, respectively.
Let $\mathbf{Y}_j \in \{0,1\}^{m_j}$ be the indicator vector for residues in $Z$. The server computes:
\[
\tilde{\mathbf{y}}_j = \frac{\bar{\mathbf{y}}_j - q_j \mathbf{1}}{p_j - q_j},
\quad
\bar{\mathbf{y}}_j = \frac{1}{n_j} \sum_{u=1}^{n_j} \mathbf{Y}_j^{(u)}.
\]
The covariance of $\tilde{\mathbf{y}}_j$ is:
\[
\Sigma_j(\mathbf f, n_j) =
\frac{1}{n_j (p_j - q_j)^2} \left( \operatorname{diag}(\boldsymbol{\pi}_j) - \boldsymbol{\pi}_j \boldsymbol{\pi}_j^\top \right).
\]

\paragraph{Global covariance.}
Define the global observation vector:
\begin{equation*}
  \tilde{\mathbf y} = 
  \begin{bmatrix} \tilde{\mathbf y}_0 \\ \vdots \\ \tilde{\mathbf y}_{\ell - 1} \end{bmatrix}
  \in \mathbb{R}^{T},
  \quad \text{with } T = \sum_{j=0}^{\ell - 1} m_j.
\end{equation*}
Since each user contributes to exactly one block, these per-block estimators are negatively correlated. For $j \ne j'$, the cross-block covariance becomes:
\begin{equation*}
  \operatorname{Cov}[\tilde{\mathbf y}_j, \tilde{\mathbf y}_{j'}]
  = -\frac{1}{n (p_j - q_j)(p_{j'} - q_{j'})}
    \boldsymbol\pi_j \boldsymbol\pi_{j'}^{\!\top}.
\end{equation*}

Stacking all components, the full covariance of $\tilde{\mathbf y}$ is:
\begin{equation*}
\begin{split}
  \Sigma(\mathbf f)
  &= \operatorname{blockdiag}\left( \mathbb{E}_{n_j}[\Sigma_j(\mathbf f, n_j)] \right)
     \,-\, \frac{\ell}{n^2} \, \mathbf{u} \mathbf{u}^\top, \\
  \text{with } \quad \mathbf{u} &=
    \begin{bmatrix}
      \boldsymbol\pi_0 / (p_0 - q_0) \\[-1pt]
      \vdots \\[-1pt]
      \boldsymbol\pi_{\ell - 1} / (p_{\ell - 1} - q_{\ell - 1})
    \end{bmatrix}.
\end{split}
\end{equation*}

\paragraph{Variance (\ie, Mean Squared Error -- MSE).}
Since the MSS estimator is unbiased (see Theorem~\ref{thm:unbiased-zero}), its mean squared error coincides with its variance: $\operatorname{MSE}_{\text{MSS}}(\mathbf f) = \operatorname{Var}[\hat{\mathbf{f}}]$.
Let $G = A_w^\dagger = (\tilde A^\top \tilde A + \lambda I)^{-1} \tilde A^\top$ be the gain matrix used in the estimator $\mathbf{\hat{f}}$. 
The variance of MSS is:
\begin{equation} \label{eq:mss-mse-general}
  \boxed{
  \operatorname{MSE}_{\text{MSS}}(\mathbf{\hat{f}})
  = \frac{1}{k} \, \operatorname{Tr} \big( G \, \Sigma(\mathbf f) \, G^\top \big)
  }.
\end{equation}
This expression holds for arbitrary distributions $\mathbf{f}$ and reflects the protocol’s total estimation risk.

\paragraph{Worst-case bound.} Let
\(
  \operatorname{MSE}_{\text{SS}}(\varepsilon,n)
  = \tfrac{4e^{\varepsilon}}{n\,(e^{\varepsilon}-1)^{2}}
\)
denote the worst-case per-coordinate MSE of a single SS block. 
Stacking the $\ell$ blocks, the MSS decoder outputs $\hat{\mathbf f}=G\tilde{\mathbf y}$ with $G:=A_{w}^{\dagger}$. 
The covariance of $\hat{\mathbf f}$ is $G\Sigma G^{\top}$, where $\Sigma$ is block-diagonal with copies of the \textsf{SS} covariance.

\begin{theorem}[Worst-Case MSE of MSS]\label{thm:mss_mse}
For any frequency vector $\mathbf f$ and any moduli choice with
finite $\kappa=\operatorname{cond}(A_{w})$,
\begin{equation} \label{eq:mss_mse_bound}
  \operatorname{MSE}_{\mathrm{MSS}}(\mathbf{\hat{f}})
  =\frac1k\operatorname{Tr}\!\bigl(G\Sigma (\mathbf{f}) G^{\!\top}\bigr)
  \;\le\;
  \boxed{\frac{4\,\kappa\,e^{\varepsilon}}%
              {n\,(e^{\varepsilon}-1)^{2}}}.
\end{equation}
\end{theorem}

\begin{proof}[Proof Sketch]
For the unbiased estimator $\hat{\mathbf f}=G\tilde{\mathbf y}$ one has $\operatorname{MSE}=k^{-1}\operatorname{Tr}\!\bigl(G\Sigma G^{\!\top}\bigr)$, where $\Sigma$ is the covariance of $\tilde{\mathbf y}$.
Trace--Cauchy--Schwarz yields $\operatorname{Tr}(G\Sigma G^{\!\top})\le \lVert G\rVert_{2}^{2}\operatorname{Tr}(\Sigma)$.
Since $G=A_{w}^{\dagger}$, $\lVert G\rVert_{2}=\sigma_{\min}^{-1}(A_{w}) = \kappa/\sigma_{\max}(A_{w})$.
Every row of $A_{w}$ contains exactly one entry 1; hence $\sigma_{\max} (A_{w})=\sqrt{S}$ with $S=\sum_j w_j$.
Cancelling $S$ gives Eq.~\eqref{eq:mss_mse_bound}.
\end{proof}

\paragraph{Bounding $\kappa$ analytically.}
Let the $\ell$ moduli $m_0,\dots,m_{\ell-1}$ be pairwise-coprime primes drawn
from the interval $[L,H]$ with $L=\frac{k}{\beta\ell}$ and $H=\frac{\beta\,k}{\ell}$ where $\beta>1$. 
Set
\begin{equation}\label{eq:t_star_alpha}
  T^{\star}
  \;=\;
  \Bigl\lceil
      \frac{\ln k}{\ln\bigl(k/(\beta\ell)\bigr)}
  \Bigr\rceil,
  \qquad
  \alpha
    = \frac{w_{\max}}{w_{\min}}
    \;\;\le\;\;
    \frac{\beta + e^{\varepsilon}}{1/\beta + e^{\varepsilon}} .
\end{equation}

\begin{theorem}[Condition-number bound]\label{thm:kappa_bound}
With probability 1 over the random prime selection,
\begin{equation}
  \kappa
  \;=\;\operatorname{cond}(A_{w})
  \;\le\;
  \alpha\,
  \frac{\ell+T^{\star}}{\ell-T^{\star}} .
  \label{eq:cond_number}
\end{equation}
Thus, any
\(
  \displaystyle
  \ell \;\ge\;
  \frac{1+\kappa_{\max}/\alpha}{\kappa_{\max}/\alpha-1}\,T^{\star}
\)
guarantees $\kappa\le\kappa_{\max}$.
\end{theorem}

\begin{proof}[Proof Sketch]
Because each column of $A$ has exactly one ``1'' per block, the Gram matrix $G=A^{\!\top}WA$ has diagonal entries $S=\sum_j w_j$ and at most $T^{\star}$ off-diagonal collisions per row. 
Gershgorin discs give $\lambda_{\min}(G)\ge w_{\min}(\ell-T^{\star})$ and $\lambda_{\max}(G)\le w_{\max}(\ell+T^{\star})$, yielding the claim.
The bound on $T^{\star}$ follows from the fact that $\prod_{j\in C(x,x')}m_j$ divides $|x-x'|<k$ for any distinct $x,x'\in[k]$.
\end{proof}

\begin{remark}[Conservative bound]\label{rmk:kappa_loose}
The lower limit \(
    \ell_{\text{theory}}
    =(1+\kappa_{\max}/\alpha)/( \kappa_{\max}/\alpha-1)\,T^{\star}
\)
is deliberately conservative.  
It combines (i) the \emph{largest} possible weight ratio \(w_{\max}/w_{\min}\) (obtained from the extremal moduli in \([L, H]\)) with (ii) Gershgorin discs that assume the \emph{maximum} number \(T^{\star}\) of off-diagonal collisions.  
Both choices over-estimate \(\kappa(A_{w})\), so the bound is sufficient but generally not tight; smaller values of~\(\ell\) frequently satisfy \(\kappa\le\kappa_{\max}\) in practice.
\end{remark}

\subsubsection{Optimized Moduli Selection}
\label{sub:mss_moduli_optimization}

The accuracy-bandwidth trade-off in \mss{} is dictated by the pairwise-coprime prime moduli $\mathbf m=(m_0,\ldots,m_{\ell-1})$. 
A valid tuple of moduli for \mss{} must:
\begin{enumerate}[label=(\roman*),leftmargin=1.7em]
  \item \textbf{cover the domain}:
        $\prod_{j=0}^{\ell-1}m_j\ge k$ (CRT property);\!
  \item \textbf{ensure full rank}:
        $\sum_{j=0}^{\ell-1}(m_j-1)\ge k$; and
  \item \textbf{yield a small condition number}
        $\kappa=\operatorname{cond}(A_w)\le\kappa_{\max}$
        so that Eq.~\eqref{eq:mss_mse_bound} guarantees low worst-case MSE.
\end{enumerate}

Because an exhaustive search is intractable, we combine the analytic
$\kappa$-bound (Theorem~\ref{thm:kappa_bound}) with lightweight random sampling.
The full pseudocode for Steps 1--3 below is provided in 
Appendix~\ref{app:moduli_selection}.

\paragraph{Step 1: Analytic lower bound for \boldmath{$\ell$}.} 
Fixing a user-defined target $\kappa_{\max}$ (we use $\kappa_{\max}=10$), Theorem~\ref{thm:kappa_bound} yields a \emph{necessary} lower limit \(
  \ell_{\text{theory}}
  =(1+\kappa_{\max}/\alpha)/( \kappa_{\max}/\alpha-1)\,T^{\star}
\).
Because this bound is loose (Remark \ref{rmk:kappa_loose}), our implementation still \emph{starts the search at $\ell_{\text{theory}}=2$} and simply discards any candidate that eventually violates $\kappa\!\le\!\kappa_{\max}$.

\paragraph{Step 2: Prime-band sampling.}
Let the user choose a search width $\beta$ (default $\beta=20$).
For each \(\ell\in\{2,\dots,\ell_{\max}\}\) we draw \(\ell\) distinct primes from \(L=k/(\beta \ell),\;H=(\beta k)/\ell\).
If coverage or rank fails, we ``bump'' random moduli to the next prime
until (i)--(ii) hold.
Sampling stops as soon as a tuple reaches \(\kappa\le\kappa_{\max}\) or
after $\#\texttt{trials}$.

\paragraph{Step 3: Moduli selection by exact MSE.}
For every candidate tuple that satisfies (i)--(iii), we compute the exact MSE in~\eqref{eq:mss_mse_bound} and keep the tuple with the smallest value, thereby selecting the communication-optimal configuration that meets the target condition number.

\paragraph{Deterministic fallback.}
If no tuple attains \(\kappa_{\max}\) in $\#\texttt{trials}$, we fall back to the first $\ell$ primes $\geq \lceil k^{1/\ell} \rceil$ and deterministically increment them (left to right) until CRT and rank conditions hold; the analytic $\kappa$-bound still applies.

\subsection{Data Reconstruction Attack on MSS}\label{sub:mss_attack}

Following recent work on adversarial analysis of LDP protocols~\cite{Gursoy2022,Arcolezi2023}, we consider a Bayesian attacker who observes a single user report $y = (J, Z)$, knows the full protocol specification, and assumes a uniform prior $\Pr[x] = 1/k$ over the domain.  

The adversary aims to infer the true user input $x$ by computing the posterior distribution and selecting the most probable value. 
The probability of a correct guess, $\Pr[\hat{x} = x]$, defines the per-message \emph{Data Reconstruction Attack} (DRA).

\paragraph{Posterior support.}
Given an MSS report $y = (j, Z)$, the attacker infers the following posterior support set:
\[
\mathcal{S}_{j,Z} = \{x \in [k] \mid x \bmod m_j \in Z\},
\]
which includes all domain elements whose residue modulo $m_j$ appears in the subset $Z$.  
Its size satisfies
\[
|\mathcal{S}_{j,Z}| \le \omega_j \cdot \left\lceil \frac{k}{m_j} \right\rceil = \omega_j C_j,
\]
where $C_j = \left\lceil k / m_j \right\rceil$ is the max number of domain values per residue.
Assuming no further knowledge, the optimal strategy is to sample uniformly from $\mathcal{S}{j,Z}$, yielding success rate $1 / |\mathcal{S}{j,Z}|$ when $x \in \mathcal{S}_{j,Z}$.

\paragraph{Upper-bound on the expected DRA.}
The attacker succeeds only if the true residue is included in $Z$, which occurs with probability $p_j$ for block $j$.
Conditioned on this, the success rate is $1 / (\omega_j C_j)$, where $C_j = \lceil k / m_j \rceil$.
To keep the analysis concise, we upper-bound the DRA by assuming the \textbf{largest possible} posterior set size $\omega_j C_j$:
\[
\mathrm{DRA}_j \;=\;
  \underbrace{p_j}_{\text{truth in }Z}\;
  \cdot\;
  \underbrace{\frac{1}{|\mathcal S_{j,Z}|}}_{\text{Bayes rule}}
  \;\le\;
  p_j\,\frac{1}{\omega_j C_j}.
\]

Averaging over the uniformly chosen block index $J$ gives the following
closed-form \emph{upper bound} on the DRA:

\begin{equation} \label{eq:mss_asr}  
\boxed{
\widehat{\mathbb E}[\mathrm{DRA}]_{\textsf{MSS}}
\;:=\;
\frac{1}{\ell}\sum_{j=0}^{\ell-1}
\frac{p_j}{\omega_j\,
          \bigl\lceil{k}/{m_j}\bigr\rceil}
}
\;\;\;\le\;\;
\mathbb E[\mathrm{DRA}]_{\textsf{MSS}}.
\end{equation}

The equality holds whenever each residue class supports the same number of
domain elements (\eg, when $m_j\mid k$), but in general the bound can be
slightly loose. 
A tight expression together with a complete proof is provided in 
Appendix~\ref{app:attack_mss}.

\section{Experimental Results} \label{sec:results}

In this section, we aim to evaluate the four aspects mentioned in Section~\ref{sec:intro}: \textit{(i) }\textbf{\textit{Utility}},\textit{ (ii)} \textbf{\textit{Communication}}, \textit{(iii)} \textbf{\textit{Server runtime}}, and \textit{(iv)} \textbf{\textit{Attackability}}.
All experiments were run on a Desktop computer with a 3.2 GHz Intel Core i9 processor, 64 GB RAM, and Python 3.11.

\paragraph{Setting.} We benchmark our \mss{} against state-of-the-art single-message frequency oracles: \pgr, \grr, \ss, and \oue{} (the optimized \rappor{} variant).
Since worst-case MSE is distribution-independent (Table~\ref{tab:ldp_summary}), we adopt the synthetic Zipf ($s=3$) and Spike ($\mathbf{f} = [1, 0, \ldots, 0]$) benchmarks from~\cite{feldman2022}, both known to induce high estimation variance.
Unless noted otherwise, we fix $n=10{,}000$ users, domain $k \in \{1{,}024, 22{,}000\}$, privacy budget $\varepsilon \in \{0.5, 1.0, \dots, 4.5, 5.0\}$, and average results over 300 independent trials.

\paragraph{Utility comparison.} Fig.~\ref{fig:mse_zipf_spike} reports the MSE of each protocol as a function of the privacy budget $\varepsilon$, under both Zipf and Spike distributions.
Notably, the relative ordering and behavior of all protocols remain consistent across Zipf and Spike distributions, confirming that our conclusions are robust to underlying data characteristics.
Among all LDP frequency-oracle protocols, GRR consistently yields the highest error, due to its $\Theta(k/e^{\varepsilon})$ scaling and lack of structure exploitation.
In contrast, OUE, SS, and PGR achieve near-optimal utility across all settings, as all three match the information-theoretic MSE bound for single-message LDP protocols.
MSS tracks SS and PGR within $\leq 1.3\times$ throughout, showing that the modular encoding adds only negligible distortion.

\begin{figure}[!htb]
  \centering
  \begin{subfigure}{0.49\linewidth}
    \centering
    \includegraphics[width=\linewidth]{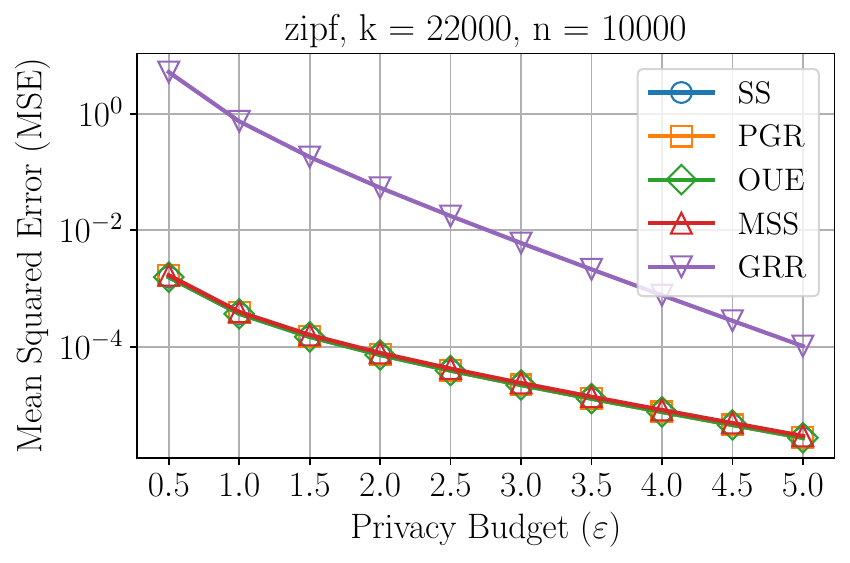}
    \caption{Zipf distribution ($s=3$).}
    \label{fig:mse_zipf}
  \end{subfigure}
  \hfill
  \begin{subfigure}{0.49\linewidth}
    \centering
    \includegraphics[width=\linewidth]{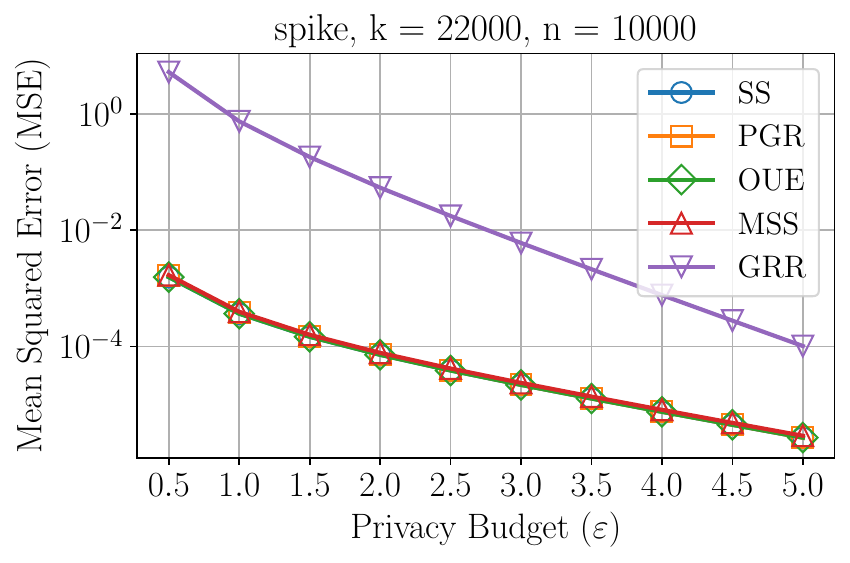}
    \caption{Spike distribution.}
    \label{fig:mse_spike}
  \end{subfigure}
  \caption{MSE vs.\ privacy parameter $\varepsilon$ for $k = 22{,}000$ and $n = 10{,}000$, under (a) Zipf and (b) Spike distributions.
  MSS closely tracks the near-optimal error curves of SS and PGR.}
  \label{fig:mse_zipf_spike}
\end{figure}

\paragraph{Communication cost.} Fig.~\ref{fig:bit_cost_k_1024_22000} shows the number of bits each user must transmit under both SS and MSS, as a function of $\varepsilon$ for $k = 1{,}024$ and $k = 22{,}000$.
Since MSS relies on a randomized moduli selection process, we report its average and standard deviation over the 300 runs.
Across all settings, MSS consistently achieves lower communication cost than SS, up to one-half in high-privacy regimes, while retaining
comparable accuracy (see Fig.~\ref{fig:mse_zipf_spike}).
We omit GRR and PGR from the figure for clarity: their per-report message length is fixed for a given $k$ (both use $O(\log k)$ bits) and is already summarized analytically in Table~\ref{tab:ldp_summary}.
GRR and PGR therefore form communication-efficient baselines, but as shown in Fig.~\ref{fig:mse_zipf_spike} and Table~\ref{tab:runtime_comparison}, they pay respectively in much higher MSE (GRR) or substantially higher decoding cost (PGR).

\begin{figure}[!htb]
  \centering
  \begin{subfigure}{0.49\linewidth}
    \centering
    \includegraphics[width=\linewidth]{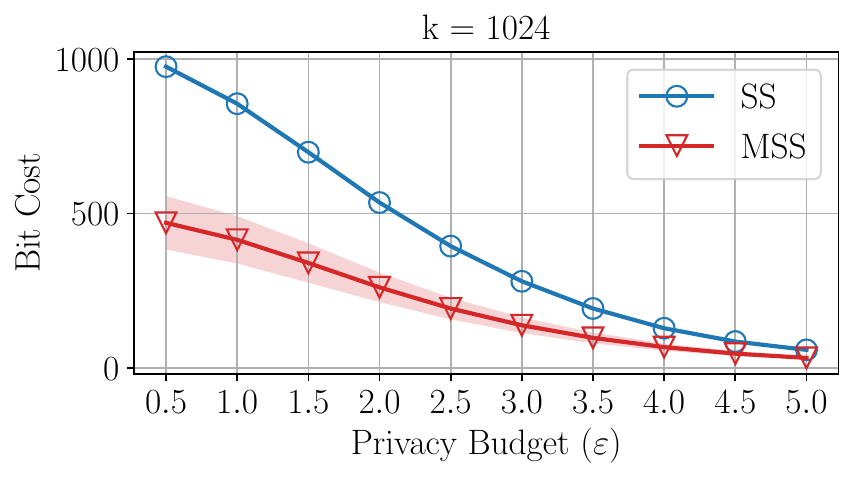}
    \caption{$k = 1{,}024$}
    \label{fig:bits_k_1024}
  \end{subfigure}
  \hfill
  \begin{subfigure}{0.49\linewidth}
    \centering
    \includegraphics[width=\linewidth]{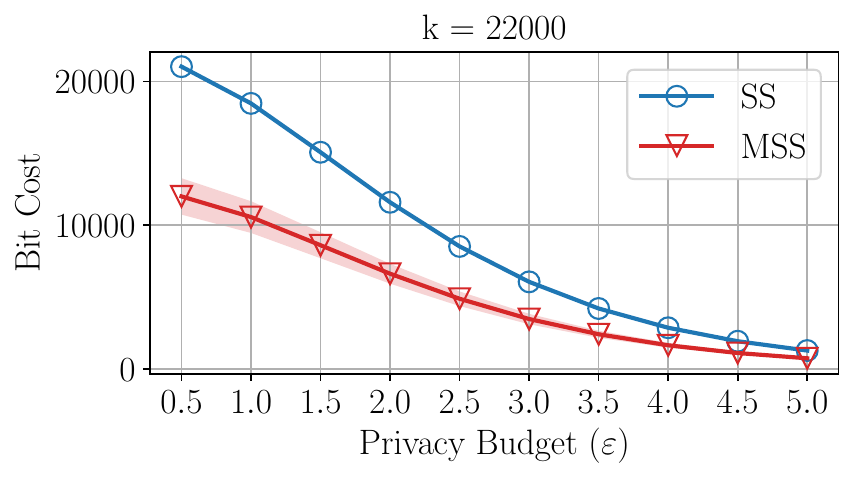}
    \caption{$k = 22{,}000$}
    \label{fig:bits_k_22000}
  \end{subfigure}
  \caption{Per-user message length (bits) of SS and MSS as a function of the $\varepsilon$, for two domain sizes. MSS consistently requires fewer bits than SS, especially in high privacy regimes.}
  \label{fig:bit_cost_k_1024_22000}
\end{figure}

\paragraph{Server runtime.}
We now compare the server-side runtime of our MSS protocol against the state-of-the-art PGR scheme by~\citet{feldman2022}. 
We run both protocols on the Zipf dataset of size $n = 10{,}000$ and domain size $k = 22{,}000$, across several privacy levels. 
Table~\ref{tab:runtime_comparison} reports the average and standard deviation of the server decoding time (in seconds) over 300 trials. 
MSS consistently outperforms PGR by large margins, achieving decoding speed-ups between $11\times$ and $448\times$.
This performance gap stems from their algorithmic differences: MSS solves a sparse weighted least-squares problem, while PGR relies on algebraic decoding over finite fields.
We do not plot GRR here, since its server cost is essentially a single histogram pass $O(n + k)$ and thus serves as a trivial lower bound on runtime; however, as Fig.~\ref{fig:mse_zipf_spike} and Fig.~\ref{fig:asr_zipf} show, GRR is not competitive in our multi-constraint regime due to its much worse utility and attackability.
The runtime spike at $\varepsilon = 4.5$ for PGR likely arises from parameter rounding and structural constraints in its projective geometry design.

\begin{table}[!htb]
    \centering
    \begin{tabular}{cccc}
    \toprule
    \multirow{2}{*}{$\mathbf{\varepsilon}$} & \multicolumn{3}{c}{\textbf{Server-Side Runtime (in seconds)}} \\    
    & MSS & PGR & MSS Speed-up\\
    \midrule
    2.0 & 0.160 $\pm$ 0.027 & 2.897 $\pm$ 0.220 & 18.1$\times$ \\
    2.5 & 0.275 $\pm$ 0.094 & 4.019 $\pm$ 0.283 & 14.6$\times$ \\
    3.0 & 0.272 $\pm$ 0.086 & 9.618 $\pm$ 0.679 & 35.4$\times$ \\
    3.5 & 0.162 $\pm$ 0.050 & 1.908 $\pm$ 0.138 & 11.7$\times$ \\
    4.0 & 0.168 $\pm$ 0.056 & 11.461 $\pm$ 0.702 & 68.3$\times$ \\
    4.5 & 0.127 $\pm$ 0.047 & 56.906 $\pm$ 3.570 & 447.8$\times$ \\
    5.0 & 0.152 $\pm$ 0.054 & 3.208 $\pm$ 0.198 & 21.1$\times$ \\
    \bottomrule
    \end{tabular}
    \caption{Average$\pm$std of server-side runtime (in seconds) for our MSS and PGR, with $k=22{,}000$ and $n=10{,}000$.
    MSS is consistently faster than PGR.
    }
    \label{tab:runtime_comparison}
\end{table}

\paragraph{Attackability.} We now evaluate the vulnerability of each LDP protocol to a single-message data reconstruction attack (DRA) (see Section~\ref{sub:mss_attack}). 
Fig.~\ref{fig:asr_zipf} shows the empirical DRA as a function of the privacy budget $\varepsilon$, under the Zipf distribution for domain sizes $k = 100$ and $k = 1{,}024$.
MSS consistently achieves the lowest DRA across all $\varepsilon$ values, confirming its robustness to reconstruction attacks.
This is due to its modular randomization strategy, which distributes the probability mass across multiple residue classes, making inference more challenging.
In contrast, GRR and SS exhibit higher attackability, especially for small $k$, as their output space is tightly linked to the input domain.
PGR behaves comparably to SS/MSS/OUE at moderate $\varepsilon$, but for larger budgets its DRA increases sharply when $k$ is \emph{smaller} than the projective-domain size $K=(q^{t}-1)/(q-1)$ required by its internal geometry.  
This truncation mismatch breaks PGR’s combinatorial symmetry and makes certain messages disproportionately informative.
A fair comparison using the non-truncated setting $k=K$ is provided in 
Appendix~\ref{app:add_results}.

\begin{figure}[!htb]
  \centering
  \begin{subfigure}{0.49\linewidth}
    \centering
    \includegraphics[width=\linewidth]{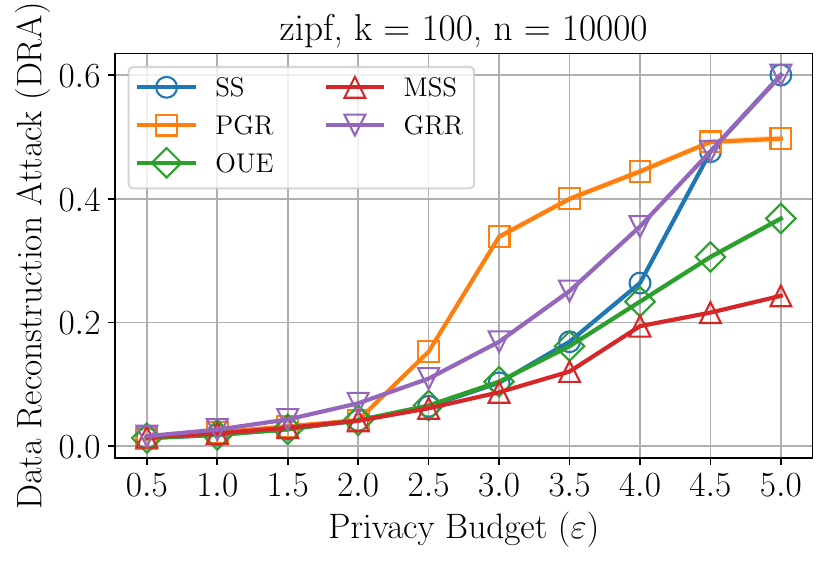}
    \caption{}
    \label{fig:asr_zipf_k_100}
  \end{subfigure}
  \hfill
  \begin{subfigure}{0.49\linewidth}
    \centering
    \includegraphics[width=\linewidth]{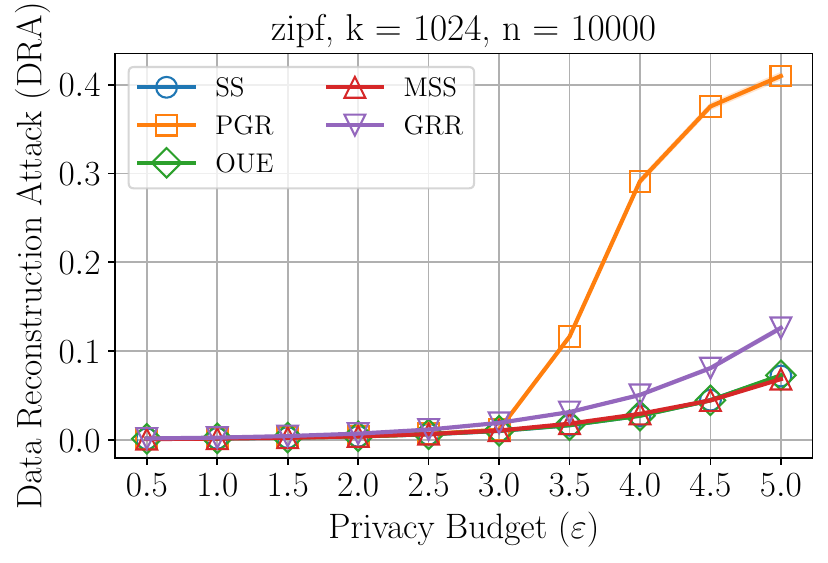}
    \caption{}
    \label{fig:asr_zipf_k_1024}
  \end{subfigure}
  \caption{Empirical Data Reconstruction Attack (DRA) of each protocol under the Zipf distribution, evaluated over $n = 10{,}000$ users.
  MSS provides the strongest protection across both small and large domains.
    }
  \label{fig:asr_zipf}
\end{figure}

\paragraph{Summary.}
Across all experiments, MSS matches the near-optimal utility of SS, OUE, and PGR, while requiring substantially fewer transmitted bits than SS (Fig.~\ref{fig:bit_cost_k_1024_22000}) and decoding orders of magnitude faster than PGR (Table~\ref{tab:runtime_comparison}). 
At the same time, MSS achieves the lowest empirical attack success rate among all protocols evaluated (Fig.~\ref{fig:asr_zipf}), demonstrating strong robustness to single-message reconstruction attacks. 
Taken together, these results position MSS in an effective operating regime for large-domain LDP frequency estimation, jointly balancing accuracy, communication cost, server-side computation, and attackability.

\paragraph{Ablation studies.}  
Appendix~\ref{app:add_results} 
presents additional experiments under different data distributions, numbers of users, and broader domain sizes, including several
ablation studies that further validate our findings.

\section{Conclusion} \label{sec:conclusion}

We introduce \mss~(MSS), a simple and powerful LDP-frequency oracle that leverages modular arithmetic to balance privacy, utility, communication, and attackability.  
Our results show that MSS achieves utility comparable to state-of-the-art protocols like SS and PGR, while significantly reducing communication cost compared to SS, lowering server runtime compared to PGR, and offering stronger protection against data reconstruction attacks.  
Future work includes extending to other statistical tasks, such as heavy hitters and multidimensional estimation.

\section*{Acknowledgments}
The author thanks Patricia Guerra-Balboa for her helpful comments on an earlier draft, and the anonymous reviewers for their insightful suggestions.
This work has been supported by the French National Research Agency (ANR): ``ANR-24-CE23-6239'' and ``ANR-23-IACL-0006''.

\bibliography{aaai2026}

\newpage

\onecolumn

\renewcommand{\thesection}{\Alph{section}}
\setcounter{section}{0}
\setcounter{secnumdepth}{2}

\section{Algorithms for Optimized Moduli Selection} \label{app:moduli_selection}

Algorithm~\ref{alg:choose-moduli} \textsc{ChooseModuli} is the \emph{outer} loop: for each block count \(\ell = 2,\dots,\ell_{\max}\), it calls \textsc{FindValidModuli} to generate valid tuples of prime moduli and selects the one with lowest estimated MSE from Eq.~\eqref{eq:mss-mse-general}.  
Algorithm~\ref{alg:find-valid-moduli} is the \emph{inner} sampler: it draws \(\ell\) distinct primes from a band centered at \(k/\ell\), repairs them until coverage and rank are satisfied, and estimates their condition number using the \emph{spectral condition number} of the weighted design matrix \(A_w\) (\ie, the ratio of its largest to smallest singular values). 
The search stops early if a sufficiently well-conditioned tuple is found; otherwise, the best candidate is retained.  
If no valid set is found, a deterministic fallback selects the first \(\ell\) primes above \(\lceil k^{1/\ell} \rceil\) and incrementally adjusts them until all constraints are satisfied.
All moduli selection is performed \emph{offline} and can be efficiently cached for reuse.
We adopt the default hyper-parameters \(\kappa_{\max}=10\), $\ell_{max}=20$, \(\beta=20\), and \(\#\texttt{trials}=10^3\).

\begin{algorithm}[H]
\caption{\textsc{ChooseModuli}$(k,\varepsilon,\ell_{\max},
                           \kappa_{\max},\beta,\#\texttt{trials})$
                           }
\label{alg:choose-moduli}
\begin{algorithmic}[1]
\Require Domain size $k$; privacy budget $\varepsilon$;
        upper bound on blocks $\ell_{\max}$;
        target condition number $\kappa_{\max}$;
        search width $\beta$;
        sampling budget $\#\texttt{trials}$
\Ensure Pairwise-coprime prime moduli $\mathbf m$
\State $(\mathbf m^\star,\mathrm{MSE}^\star)\gets(\textsc{None},\infty)$
\For{$\ell=2$ \textbf{to} $\ell_{\max}$} \Comment{Remark~\ref{rmk:kappa_loose}}
    \State $\mathbf m\gets$\textsc{FindValidModuli}%
           $(k,\ell,\kappa_{\max},\beta,\#\texttt{trials})$
    \If{$\mathbf m=\textsc{None}$} \textbf{continue} \EndIf
    \State Evaluate MSE with Eq.~\eqref{eq:mss-mse-general}
    \If{MSE $<\mathrm{MSE}^\star$}
        \State $(\mathbf m^\star,\mathrm{MSE}^\star)
               \gets(\mathbf m,\mathrm{MSE})$
    \EndIf
\EndFor
\State \Return $\mathbf m^\star$
\end{algorithmic}
\end{algorithm}

\vspace{0.4em}

\begin{algorithm}[H]
\caption{\textsc{FindValidModuli}$(k,\ell,\kappa_{\max},
                                   \beta,\#\texttt{trials})$
                                   }
\label{alg:find-valid-moduli}
\begin{algorithmic}[1]
\Require Domain size $k$; block count $\ell$;
        target condition number $\kappa_{\max}$;
        search width $\beta$;
        sampling budget $\#\texttt{trials}$
\Ensure Valid moduli tuple $\mathbf m$ or \textsc{None}
\State $L\gets k/(\beta\ell)$,\;
       $H\gets\min\!\bigl(\beta k/\ell,\;0.95\,k\bigr)$
\State $\mathcal P\gets\textsc{PrimesInBand}(L,H)$ \Comment{all primes in $[L,H]$}
\For{$t=1$ \textbf{to} $\#\texttt{trials}$}
    \State Sample $\ell$ distinct primes $\mathbf m\subset\mathcal P$
    \While{$\prod_j m_j<k$ \textbf{or} $\sum_j(m_j-1)<k$}
        \State Bump a random $m_j$ to next prime $>m_j$
    \EndWhile
    \State $\kappa\gets \operatorname{cond}(A_w)$  \Comment{Spectral condition number (largest/smallest singular value)}
    \If{$\kappa\le\kappa_{\max}$} \Return $\mathbf m$ \EndIf
\EndFor
\Statex \hrulefill
\Statex \textbf{Deterministic Fallback:}
\State Initialize $\mathbf m$ with first $\ell$ primes $\ge \lceil k^{1/\ell} \rceil$
\State $i \gets 0$
\While{$\prod_j m_j<k$ \textbf{or} $\sum_j(m_j{-}1)<k$}
    \State $m_i \gets$ next prime $> m_i$
    \State $i \gets (i+1)\bmod \ell$
\EndWhile
\State $\kappa\gets \operatorname{cond}(A_w)$
\If{$\kappa\le\kappa_{\max}$} \Return $\mathbf m$ \Else\ \Return \textsc{None} \EndIf
\end{algorithmic}
\end{algorithm}

\section{Expected Data Reconstruction Attack on MSS} \label{app:attack_mss}

\paragraph{Setting.}
Let the domain be $[k]=\{0,\dots,k-1\}$.
MSS operates by selecting a random index $j \in [\ell]$ and computing the residue $x \bmod m_j$ for a fixed input $x \in [k]$.  
It then applies the SS mechanism~\cite{wang2016mutual} over the domain $[m_j]$ to perturb the residue: the true value is included in the output set $Z \subset [m_j]$ of size $\omega_j = \lfloor m_j / (e^\varepsilon + 1) \rceil$ with probability
\[
p_j = \frac{\omega_j e^{\varepsilon}}{\omega_j e^{\varepsilon} + m_j - \omega_j},
\]
and the remaining $\omega_j - 1$ elements are drawn uniformly without replacement from the other $m_j - 1$ residues.  
All elements in $Z$ are therefore equally likely to be the true residue from the attacker's perspective.
The user reports $(j, Z)$ to the server.

We analyze a Bayesian attacker who (i) knows $k, \mathbf m, \varepsilon$, (ii) observes one report $y=(j,Z)$, and (iii) assumes a uniform prior $\Pr[x]=1/k$.

\paragraph{Residue multiplicity.}
Write the Euclidean division of $k$ by $m_j$ as $k = e_j m_j + r_j$ with
\[
e_j = \left\lfloor \frac{k}{m_j} \right\rfloor, \qquad 0 \le r_j < m_j.
\]
The number of domain values mapping to each residue is:
\begin{equation}\label{eq:DRA-multiplicity}
  n_{j,z}\;=\;
  \bigl|\{x\in[k]:x\bmod m_j=z\}\bigr|
  \;=\;
  \begin{cases}
    e_j+1,&z<r_j,\\
    e_j,  &z\ge r_j.
  \end{cases}
\end{equation}
Hence $\sum_{z=0}^{m_j-1} n_{j,z}=k$.

\paragraph{Posterior support for a fixed report.}
Given a report $y=(j,Z)$, the posterior support of $x$ is
\[
\mathcal S_{j,Z} = \{x \in [k] : x \bmod m_j \in Z\}.
\]
If the true residue $r$ is not in $Z$, the attacker fails.  
Otherwise, the subset takes the form $Z = \{r\} \cup S$, where $S \subset [m_j] \setminus \{r\}$ and $|S| = \omega_j - 1$.  
The support size is
\begin{equation}\label{eq:DRA-support-size}
  |\mathcal S_{j,Z}| = n_{j,r} + \underbrace{\sum_{u \in S} n_{j,u}}_{=:T_{j,r}}.
\end{equation}

Because the attacker guesses uniformly from the posterior support, its success probability is the reciprocal of this random support size.

\paragraph{Conditional success for residue $r=z$.}
Fix block $j$ and suppose the true residue is $r = z$.
The filler set $S$ is drawn uniformly without replacement from the remaining residues, which makes the filler weight random:
\begin{equation}\label{eq:DRA-T-def}
  T_{j,z} = \sum_{u \in S} n_{j,u}.
\end{equation}
The expectation below is taken over the randomness of the filler set $S$.
Thus, using Eq.~\eqref{eq:DRA-support-size}, the success probability is:
\begin{equation}\label{eq:DRA-cond}
  \Pr[\hat x = x \mid J = j, r = z]
  = p_j \cdot \mathbf E\left[ \frac{1}{n_{j,z} + T_{j,z}} \right].
\end{equation}

\paragraph{Expected accuracy for one block.}
Weighting Eq.~\eqref{eq:DRA-cond} by the marginal probability $\Pr[r = z \mid J = j] = n_{j,z}/k$ yields:
\begin{equation}\label{eq:DRA-block}
  \mathrm{DRA}_j
  =
  p_j
  \sum_{z=0}^{m_j-1}
    \frac{n_{j,z}}{k}
    \cdot \mathbf E\left[ \frac{1}{n_{j,z} + T_{j,z}} \right].
\end{equation}

\paragraph{Global expected accuracy.}
Averaging over the uniformly chosen block index $J \in [\ell]$ gives the total expected DRA success rate:
\begin{equation}\label{eq:DRA-exact}
  \boxed{%
    \mathbf E[\mathrm{DRA}]_{\textsf{MSS}}
    =
    \frac{1}{\ell}
    \sum_{j=0}^{\ell-1}
      p_j
      \sum_{z=0}^{m_j-1}
        \frac{n_{j,z}}{k}
        \cdot \mathbf E\left[ \frac{1}{n_{j,z} + T_{j,z}} \right]
  }
\end{equation}

\paragraph{Special cases and upper bound.}
\begin{enumerate}[label=(\alph*)]
\item \textbf{Equal-size residues} ($m_j \mid k$).  
      In this case, $n_{j,z} = k/m_j$ and $T_{j,z}$ is deterministic, so
      \[
      \mathrm{DRA}_j = \frac{p_j}{\omega_j \cdot \lceil k / m_j \rceil}.
      \]
      The bound in Eq.~\eqref{eq:mss_asr} is tight.

\item \textbf{Singleton subsets} ($\omega_j = 1$).  
      Then $T_{j,z} = 0$ and
      \[
      \mathbf E[\mathrm{DRA}]_{\textsf{MSS}} = \frac{1}{\ell k} \sum_j p_j \cdot \min(m_j, k).
      \]

\item \textbf{Upper-bound in the main text.}  
      Using Jensen’s inequality with the convex function $x \mapsto 1/x$, we obtain
      \[
      \mathbf E\left[ \frac{1}{n_{j,z} + T_{j,z}} \right] \ge \frac{1}{\omega_j \cdot \lceil k / m_j \rceil},
      \]
      \noindent which turns Eq.~\eqref{eq:DRA-block} into the concise upper bound reported in Section~\ref{sub:mss_attack}.
\end{enumerate}

\noindent
Eq.~\eqref{eq:DRA-exact} therefore provides the exact expected single-report success rate for the Bayesian attacker under a uniform prior and uniform guessing over the posterior support.  
The main paper Eq.~\eqref{eq:mss_asr} serves as a conservative, closed-form upper bound.

\section{Expected Data Reconstruction Attack on PGR}
\label{app:attack_pgr}

\paragraph{Setting.}
Let the domain be $[k] = \{0,\dots,k - 1\}$.
The \pgr{} (PGR)~\cite{feldman2022} mechanism embeds $[k]$ into the projective space
\[
K = \frac{q^t - 1}{q - 1},
\]
where $q$ is a prime power and $t$ is the smallest integer such that $K \ge k$.
Each element $x \in [k]$ is mapped to a canonical projective vector $v_x \in \mathbb{F}_q^t$.
For each $x$, the \emph{preferred set} $S(x)$ is the set of projective points $y \in [K]$ whose canonicalized vectors are orthogonal to $v_x$:
\[
S(x) = \{\, y \in [K] : \langle v_x, v_y \rangle_q = 0 \,\}.
\]
All preferred sets have the same size
\[
c_{\mathrm{set}} = \frac{q^{t-1} - 1}{q - 1}.
\]

Under privacy parameter $\varepsilon$, the mechanism outputs $Y=y$ with probabilities
\[
\Pr[Y = y \mid X = x] =
\begin{cases}
e^{\varepsilon} p, & y \in S(x), \\[1mm]
p,                 & y \notin S(x),
\end{cases}
\qquad
p = \dfrac{1}{(e^\varepsilon - 1)c_{\mathrm{set}} + K}.
\]
We analyze a Bayesian attacker assuming a uniform prior $\Pr[X=x] = 1/k$ and observing a single report $Y=y$.

\paragraph{Posterior support for a fixed report.}
For a fixed message $y \in [K]$, define the set of consistent inputs
\[
A(y) = \{ \, x \in [k] : y \in S(x) \,\}.
\]
If $A(y) = \varnothing$, then $\Pr[X = x \mid Y = y] = 1/k$ for all $x$.
If $A(y) \neq \varnothing$, Bayes' rule gives
\[
\Pr[X = x \mid Y = y]
=
\begin{cases}
\alpha_y, & x \in A(y), \\[1mm]
\beta_y,  & x \notin A(y),
\end{cases}
\]
where
\[
\alpha_y
=
\dfrac{e^\varepsilon}{k + (e^\varepsilon - 1)\,|A(y)|},
\qquad
\beta_y
=
\dfrac{1}{k + (e^\varepsilon - 1)\,|A(y)|}.
\]

The Bayes-optimal single-message attacker succeeds with
\begin{equation}
\label{eq:pgr_conditional_asr}
\Pr[\hat X = X \mid Y = y]
=
\begin{cases}
\frac{1}{k}, &
|A(y)| = 0,\\[2mm]
\displaystyle
\frac{e^\varepsilon}
     {k + (e^\varepsilon - 1)\,|A(y)|}, &
|A(y)| > 0.
\end{cases}
\end{equation}

\paragraph{Distribution of messages.}
By symmetry of PGR and the uniform prior, all messages are equally likely:
\[
\Pr[Y = y] = \frac{1}{K} \qquad \forall y \in [K].
\]

\paragraph{Exact expected DRA accuracy.}
Let
\[
N_{\mathrm{pref}}
=
|\{\, y \in [K] : |A(y)| > 0 \,\}|
\]
be the number of messages that have at least one preferred input in $[k]$.
Then
\begin{equation}
\label{eq:pgr_exact_asr}
\boxed{
\mathbf{E}[\mathrm{DRA}]_{\mathrm{PGR}}
=
\frac{1}{K}
\left(
\sum_{y : |A(y)| > 0}
\frac{e^\varepsilon}{k + (e^\varepsilon - 1)\,|A(y)|}
\;+\;
\frac{K - N_{\mathrm{pref}}}{k}
\right)
}
\end{equation}
where $|A(y)| = |\{\, x < k : \langle v_x, v_y \rangle_q = 0 \,\}|$.

\paragraph{Special case: full projective domain.}
When $k = K$, projective symmetry implies $|A(y)| = c_{\mathrm{set}}$ for all messages $y$, yielding
\begin{equation}
\label{eq:pgr_closed_form}
\boxed{
\mathbf{E}[\mathrm{DRA}]_{\mathrm{PGR(full)}}
=
\frac{e^\varepsilon}
     {K + (e^\varepsilon - 1)c_{\mathrm{set}}}.
}
\end{equation}
This closed form applies only when the domain is \emph{not} truncated ($k = K$).
For $k < K$, Eq.~\eqref{eq:pgr_exact_asr} must be used.

\section{Complementary Results} \label{app:add_results}

\paragraph{Ablation: Analytical vs. Empirical MSE.}
Fig.~\ref{fig:mse_theory_empirical} compares the analytical MSE derived for both SS and MSS against their empirical counterparts, under both Zipf and Spike distributions for $k \in \{1{,}024, 22{,}000\}$. 
The analytical expressions closely match the empirical observations across all regimes of $\varepsilon$, consistently tracking the empirical trends. 
This validates the tightness and reliability of our closed-form MSE derivation for MSS (Eq.~\eqref{eq:mss-mse-general}) in practice, regardless of the underlying data distribution or domain size. 
The results provide strong evidence that our theoretical analysis generalizes well and can be used for practical design decisions without the need for repeated empirical tuning.

\begin{figure*}[!htb]
  \centering
  \begin{subfigure}{0.98\linewidth}
    \centering
    \includegraphics[width=\linewidth]{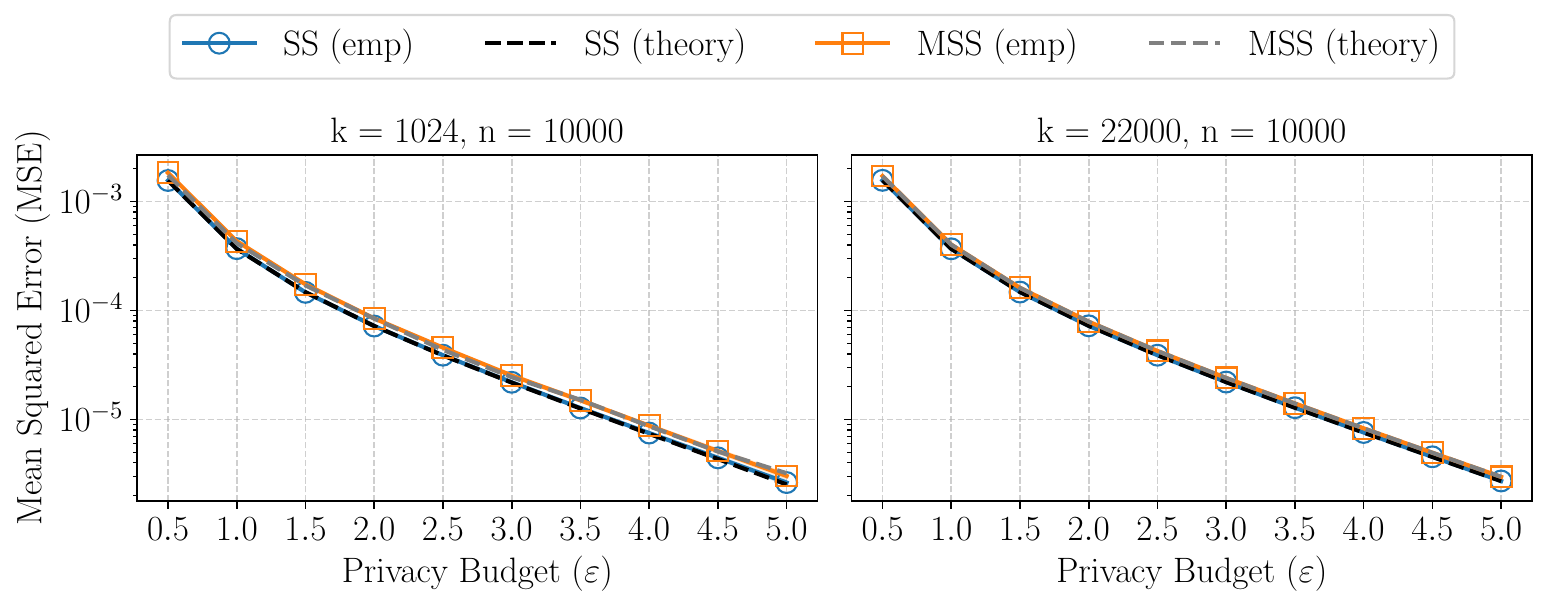}
    \caption{Zipf distribution ($s=3$).}
    \label{fig:mse_theory_zipf}
  \end{subfigure}\\
  \hfill
  \begin{subfigure}{0.98\linewidth}
    \centering
    \includegraphics[width=\linewidth]{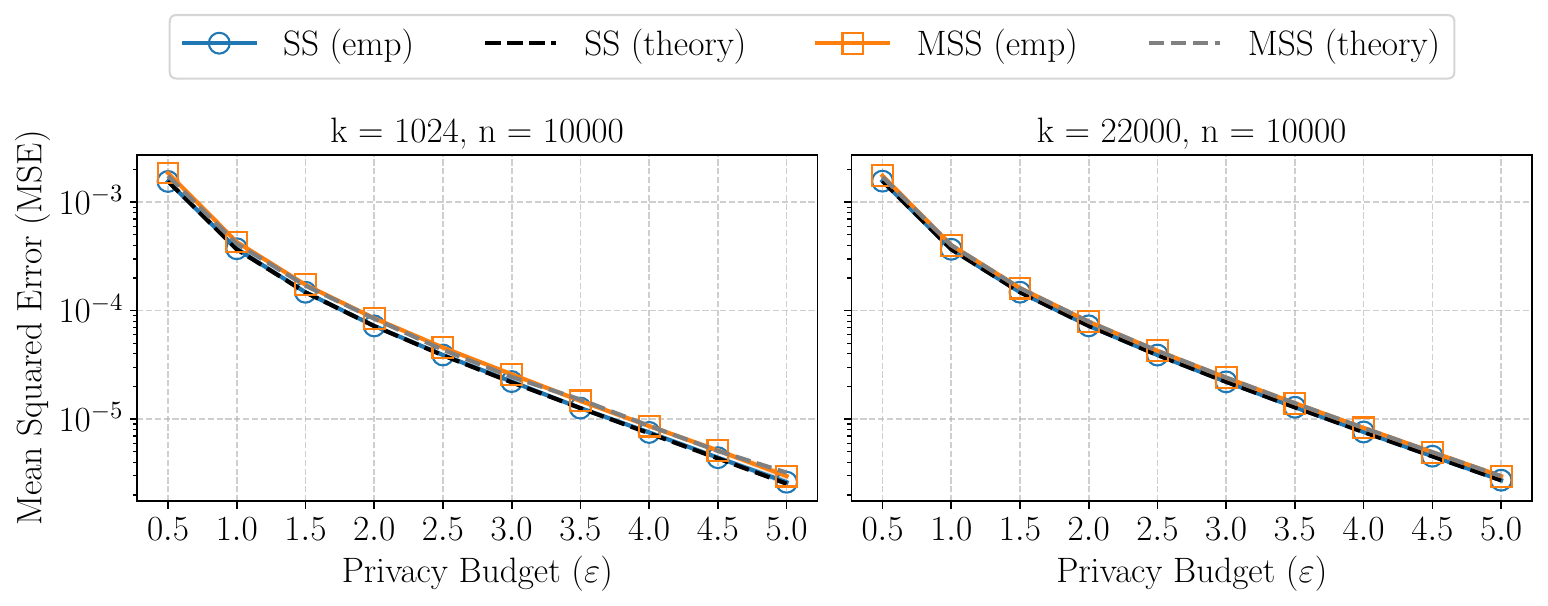}
    \caption{Spike distribution.}
    \label{fig:mse_theory_spike}
  \end{subfigure}
  \caption{
    Comparison between analytical and empirical MSE for SS and MSS protocols, across a range of privacy budgets $\varepsilon$, for $k = 1{,}024$ and $k = 22{,}000$ under both Zipf and Spike distributions.
    Each empirical MSE is averaged over 300 runs, while analytical curves are computed in closed-form expressions.
  }
  \label{fig:mse_theory_empirical}
\end{figure*}

\paragraph{Ablation: Sensitivity to $\ell$ and MSS[OPT].}
Fig.~\ref{fig:mss_ell_ablation} presents an ablation study analyzing how the performance of the MSS protocol varies with fixed numbers of moduli $\ell \in \{3, 6, 9, 12, 15\}$, compared to the analytically optimized MSS[OPT].
We observe that the relationship between $\ell$, utility (MSE), and communication cost (bit cost) is non-linear: in some regimes, intermediate values such as $\ell=9$ can outperform both smaller ($\ell=3$) and larger ($\ell=12$) settings in terms of accuracy and communication efficiency.
This highlights the complex trade-offs induced by modular encoding, where increasing $\ell$ does not guarantee monotonic improvements.
Notably, MSS[OPT] consistently selects a configuration that achieves near-optimal utility across all privacy budgets, confirming the effectiveness of our moduli selection strategy.
It is important to emphasize that our optimization procedure (Section~\ref{sub:mss_moduli_optimization} and Appendix~\ref{app:moduli_selection}) currently targets minimizing analytical MSE only.
This objective implicitly balances communication and robustness in many settings, but it is not guaranteed to yield optimal trade-offs across all criteria.
Future work could extend this framework to support multi-objective optimization, for instance, incorporating additional metrics such as bit cost or attackability, enabling more fine-grained control over privacy-utility-efficiency trade-offs.

\begin{figure*}[!htb]
  \centering
  \includegraphics[width=\linewidth]{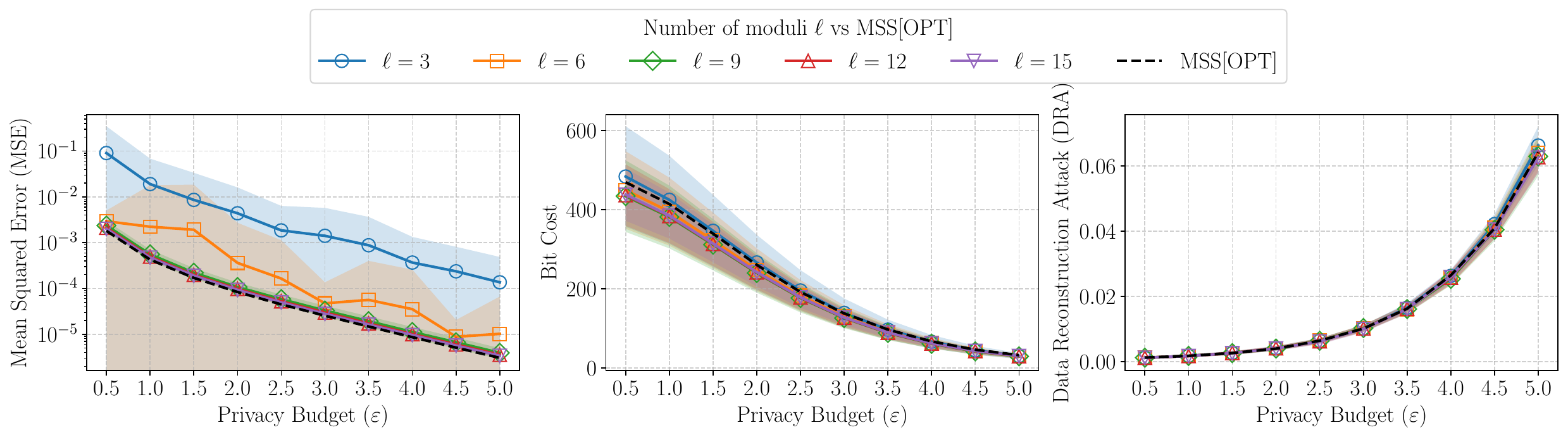}
  \caption{
    Ablation study showing the impact of the number of moduli $\ell \in \{3, 6, 9, 12, 15\}$ on the utility (left), communication cost (middle), and attackability (right) of the MSS protocol.
    The dashed black curve represents the performance of our \mss{} protocol (\ie, MSS[OPT]), which automatically selects $\ell$ and the moduli via our analytical optimization procedure.
    Results are averaged over 300 runs for $k=1{,}024$, under the Zipf distribution.
  }
  \label{fig:mss_ell_ablation}
\end{figure*}

\paragraph{Ablation: Empirical vs. Analytical DRA.}
To evaluate the tightness of our analytical DRA derivations (Section~\ref{sub:mss_attack} and Appendix~\ref{app:attack_mss}), we compare theoretical and empirical data reconstruction attack (DRA) for both SS and MSS across privacy budgets.
As shown in Fig.\ref{fig:asr_theory_vs_empirical}, SS exhibits a near-perfect match between analytical and empirical DRA.
For MSS, however, the analytical DRA bound consistently overestimates the true attackability.
This gap is theoretically expected: the analytical expression used in the main paper is a conservative upper bound derived using Jensen’s inequality over the random support size of the modular posterior (see Appendix~\ref{app:attack_mss}).
The underlying randomness of the filler set in each residue block, and the multiplicity of values per residue, makes the exact computation of MSS’s DRA more intricate, resulting in a looser but guaranteed-safe upper bound.
These results show that while the MSS DRA bound is not tight, it still provides a safe analytical proxy and reinforces that MSS offers strong practical protection against reconstruction attacks.

\begin{figure*}[!htb]
  \centering
  \begin{subfigure}{0.49\linewidth}
    \centering
    \includegraphics[width=\linewidth]{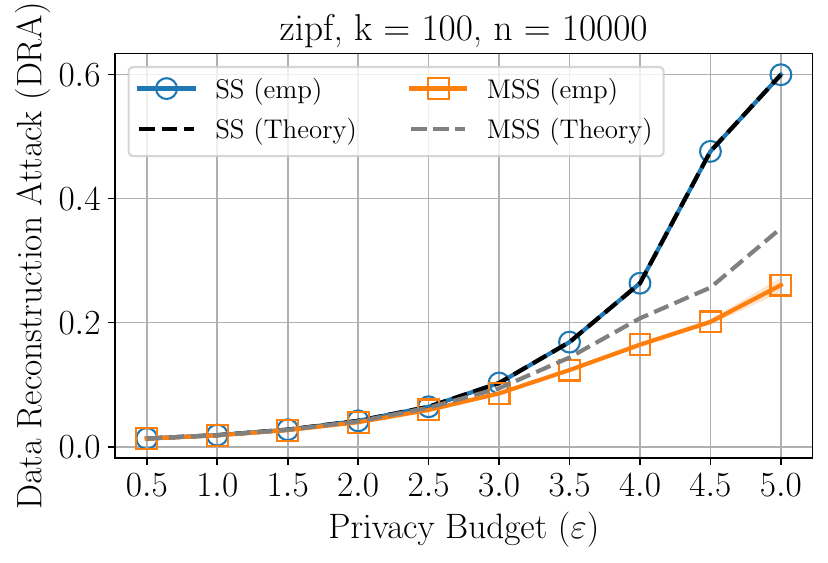}
    \caption{}
    \label{fig:asr_theory_zipf_k100}
  \end{subfigure}
  \hfill
  \begin{subfigure}{0.49\linewidth}
    \centering
    \includegraphics[width=\linewidth]{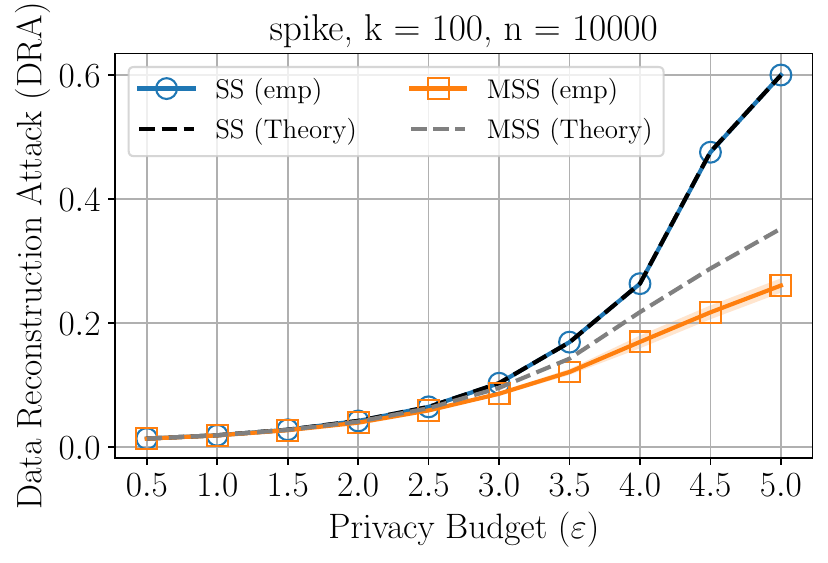}
    \caption{}
    \label{fig:asr_theory_spike_k100}
  \end{subfigure}
  \medskip
  \begin{subfigure}{0.49\linewidth}
    \centering
    \includegraphics[width=\linewidth]{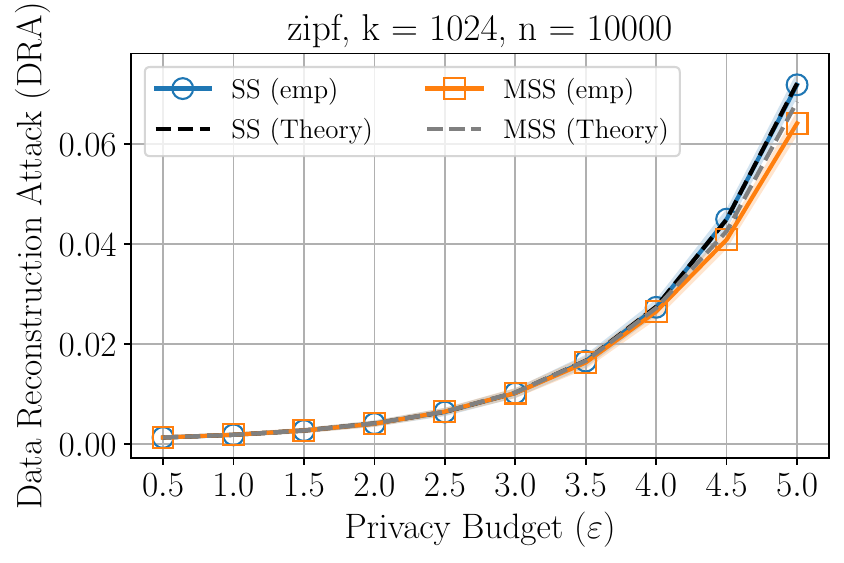}
    \caption{}
    \label{fig:asr_theory_zipf_k1024}
  \end{subfigure}
  \hfill
  \begin{subfigure}{0.49\linewidth}
    \centering
    \includegraphics[width=\linewidth]{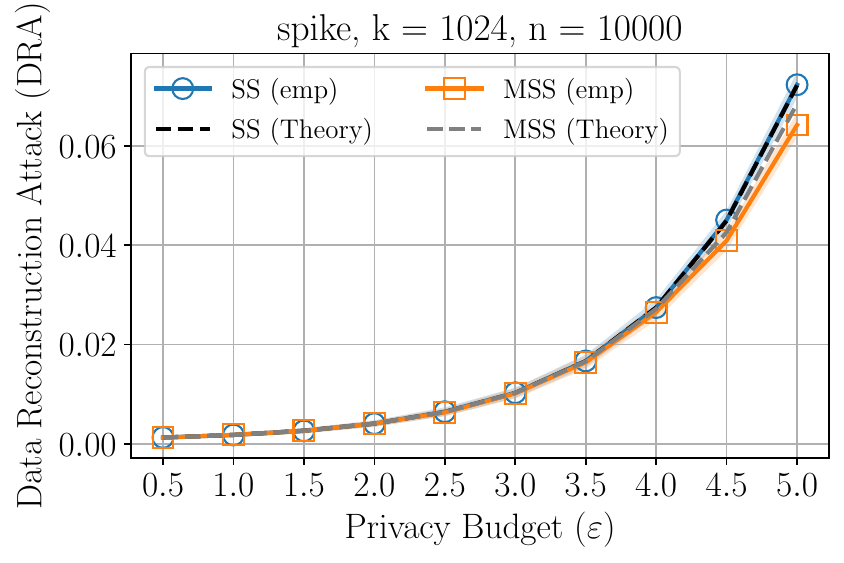}
    \caption{}
    \label{fig:asr_theory_spike_k1024}
  \end{subfigure}
  \caption{
  Empirical vs. analytical data reconstruction attack (DRA) under Zipf and Spike distributions, for both small and large domains.
  While SS closely matches its analytical DRA, MSS shows a consistent gap, confirming that the bound is conservative.
  }
  \label{fig:asr_theory_vs_empirical}
\end{figure*}

\paragraph{Additional Results: Utility comparison.} To complement the results reported in the main paper (Fig.~\ref{fig:mse_zipf_spike}), we include additional empirical analyses in Fig.~\ref{fig:error_distribution_complementary}. 
This figure presents CDFs of estimation error for both Zipf and Spike distributions, alongside extended MSE \emph{vs} $\varepsilon$ plots under various settings.
These additional plots confirm and extend the conclusions drawn in the main paper. 
The CDFs in subfigures (a)--(d) show that MSS consistently yields estimation errors close to SS and PGR, with low variance across seeds and strong robustness to the underlying data distribution (Zipf or Spike). 
Subfigures (e)--(h) further illustrate that the utility trends reported in the main text also hold at $k=1{,}024$, validating the generality of our findings. 
Across all configurations, GRR maintains the highest error. 
MSS remains competitive against the best protocols (OUE, SS, and PGR) in terms of MSE.

\begin{figure*}[!htb]
  \centering

  \begin{subfigure}{0.49\linewidth}
    \centering
    \includegraphics[width=\linewidth]{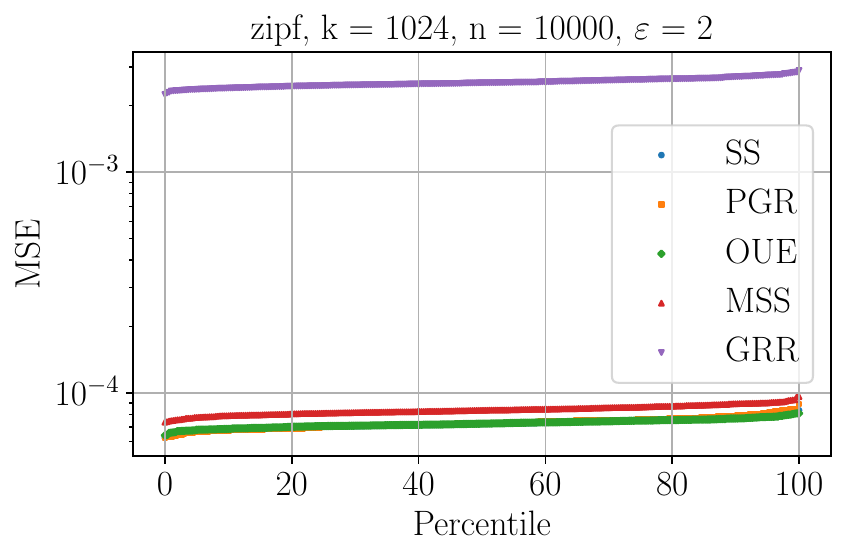}
    \caption{CDF under Zipf ($s=3$), $k = 1{,}024$.}
  \end{subfigure}
  \hfill
  \begin{subfigure}{0.49\linewidth}
    \centering
    \includegraphics[width=\linewidth]{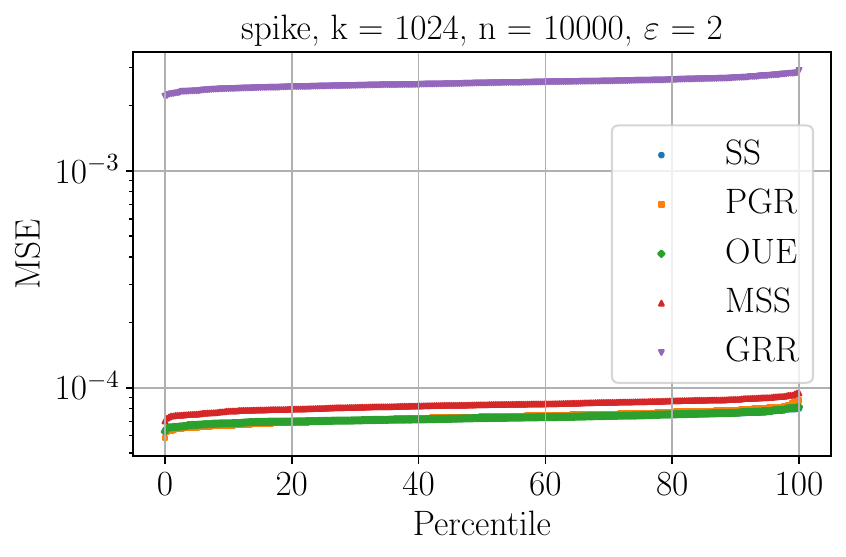}
    \caption{CDF under Spike, $k = 1{,}024$.}
  \end{subfigure}

  \begin{subfigure}{0.49\linewidth}
    \centering
    \includegraphics[width=\linewidth]{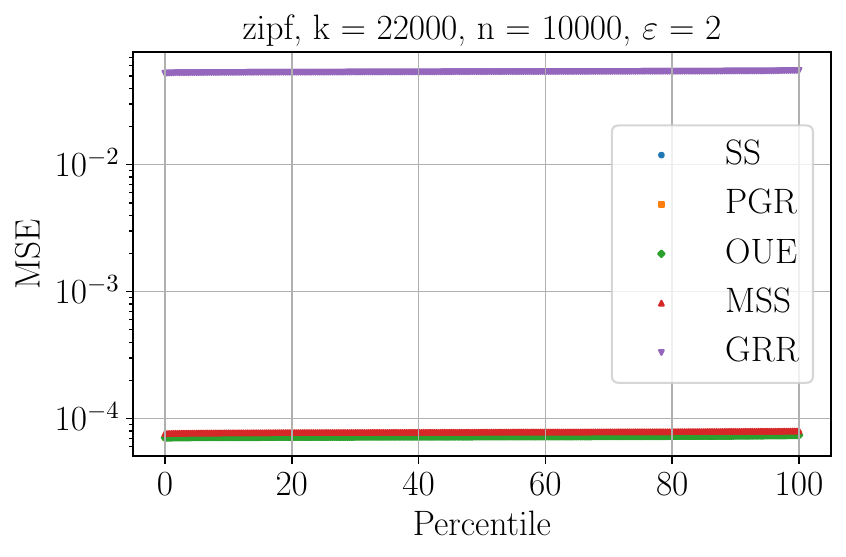}
    \caption{CDF under Zipf ($s=3$), $k = 22{,}000$.}
  \end{subfigure}
  \hfill
  \begin{subfigure}{0.49\linewidth}
    \centering
    \includegraphics[width=\linewidth]{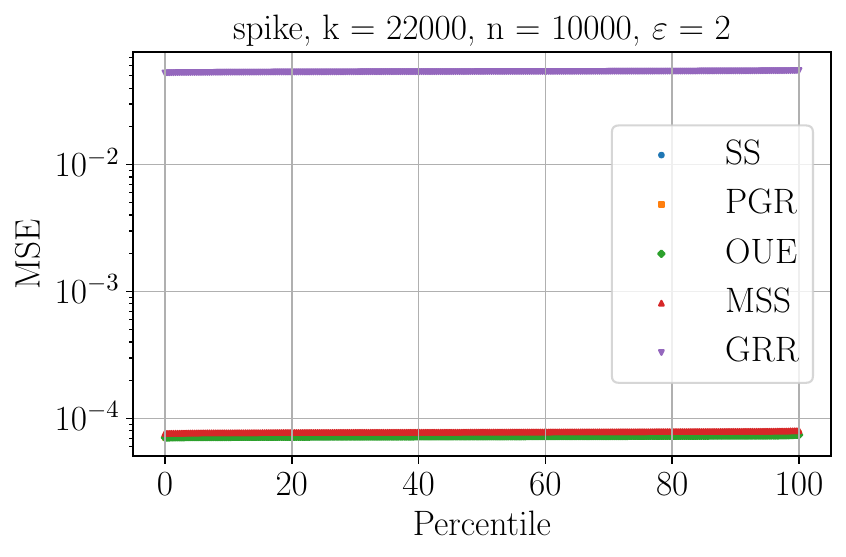}
    \caption{CDF under Spike, $k = 22{,}000$.}
  \end{subfigure}

  \begin{subfigure}{0.49\linewidth}
    \centering
    \includegraphics[width=\linewidth]{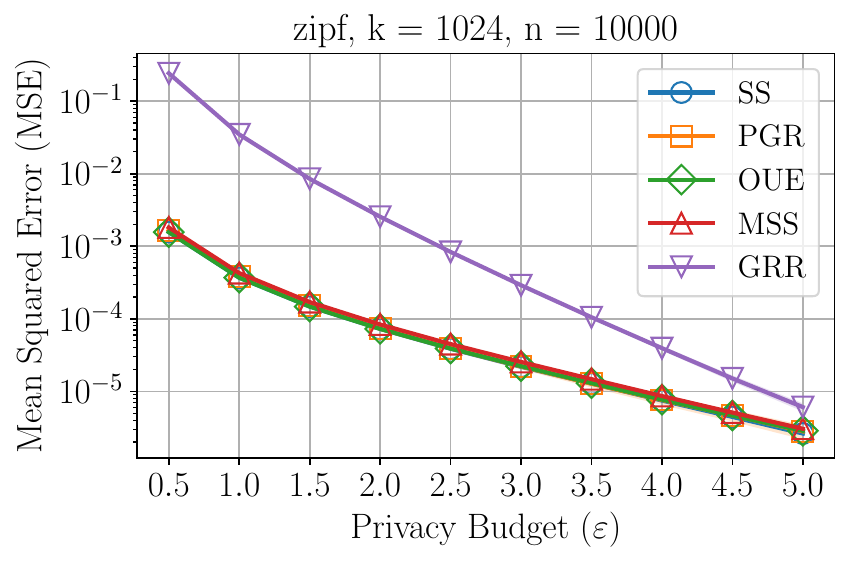}
    \caption{MSE vs $\varepsilon$ under Zipf ($s=3$), $k = 1{,}024$.}
  \end{subfigure}
  \hfill
  \begin{subfigure}{0.49\linewidth}
    \centering
    \includegraphics[width=\linewidth]{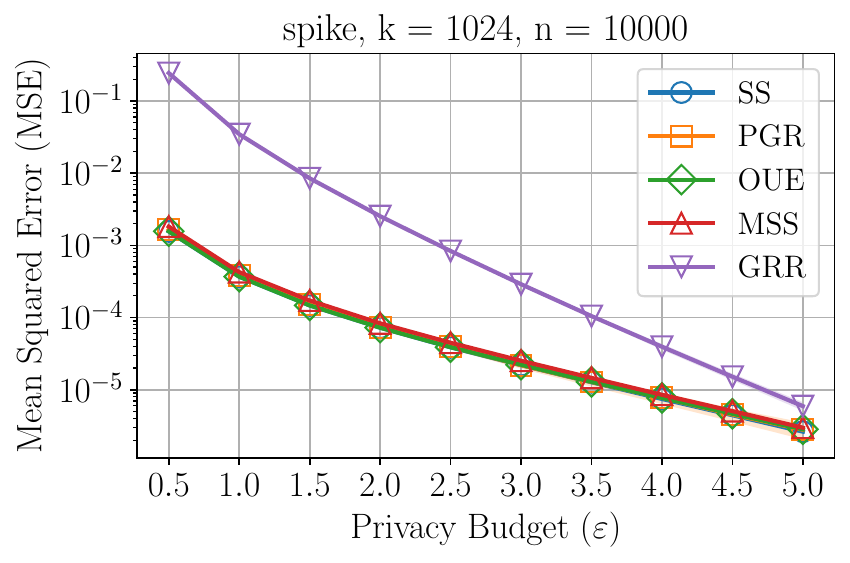}
    \caption{MSE vs $\varepsilon$ under Spike, $k = 1{,}024$.}
  \end{subfigure}

  \caption{
    \textbf{Error distribution from experiments.} 
    Subfigures (a)--(d) show the CDFs of the estimation error (MSE) under 300 runs for Zipf and Spike distributions at two domain sizes. 
    Subfigures (e)--(f) show the variation of MSE with $\varepsilon$ under Zipf and Spike distributions for $k = 1{,}024$.
  }
  \label{fig:error_distribution_complementary}
\end{figure*}

\paragraph{Additional Results: Attackability.}
To complement our Zipf-based results in the main paper (Fig.~\ref{fig:asr_zipf}), we evaluate the empirical attackability of each protocol under the Spike distribution.
As shown in Fig.~\ref{fig:asr_spike}, the trends are consistent with those observed for Zipf: MSS maintains the lowest data reconstruction attack (DRA) across all privacy budgets $\varepsilon$ and domain sizes.
This consistency confirms that the robustness of MSS to data reconstruction attacks is largely independent of the input distribution.
In contrast, protocols like GRR and SS remain more susceptible to attack due to their direct encoding of input values.
These results reinforce that MSS's modular design offers strong defense against single-message attacks, regardless of how the data is distributed.

While these trends hold generally, it is important to highlight a distinctive behavior of PGR: unlike other mechanisms, PGR's internal domain size is determined by the projective geometry induced by the privacy budget $K(\varepsilon) = (q^t - 1)/(q - 1)$ with $q \approx e^{\varepsilon}+1$.
To isolate the impact of this dependence, we perform an additional ablation (Fig.~\ref{fig:asr_pgr_aligned}) in which, for each privacy level $\varepsilon$, we set the evaluation domain to the corresponding projective size $K(\varepsilon)$ and run \emph{all} protocols on this shared domain. 
We repeat this ablation using $\mathrm{base\_k}=100$ and $1{,}024$, to confirm that the phenomenon is consistent across both small and large domains.
This geometry-aligned setting removes the truncation mismatch that affects PGR when $k < K(\varepsilon)$ and places all baselines in the domain naturally required by PGR.
As shown in Fig.~\ref{fig:asr_pgr_aligned}, once this alignment is enforced, PGR’s attackability no longer exhibits the sharp rise seen in Figs.~\ref{fig:asr_zipf} and~\ref{fig:asr_spike}, while MSS continues to offer the lowest DRA across all~$\varepsilon$.

\begin{figure*}[!htb]
  \centering
  \begin{subfigure}{0.49\linewidth}
    \centering
    \includegraphics[width=\linewidth]{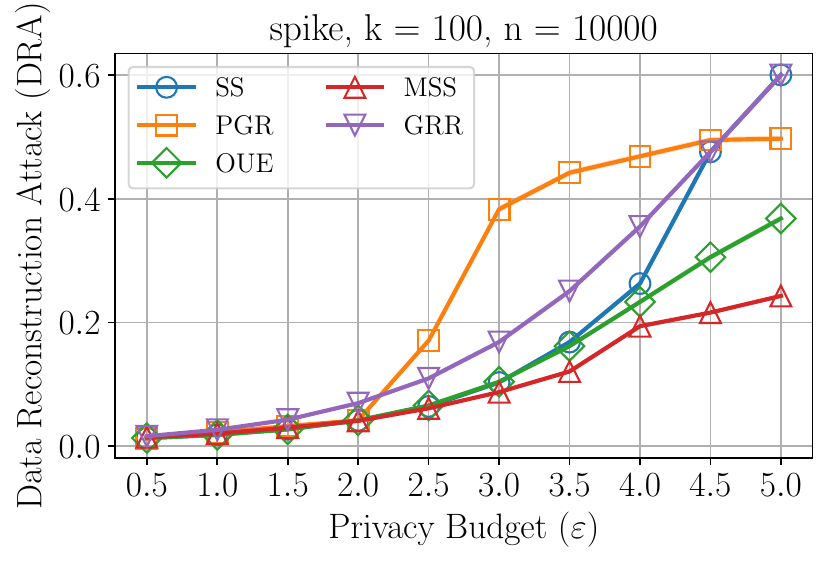}
    \caption{Spike, $k = 100$.}
    \label{fig:asr_ablation_spike_k_100}
  \end{subfigure}
  \hfill
  \begin{subfigure}{0.49\linewidth}
    \centering
    \includegraphics[width=\linewidth]{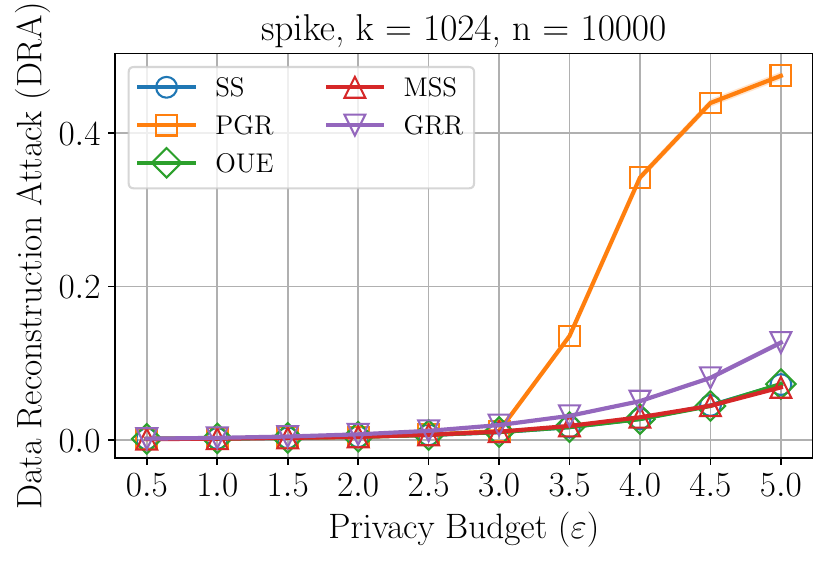}
    \caption{Spike, $k = 1{,}024$.}
    \label{fig:asr_ablation_spike_k_1024}
  \end{subfigure}
  \caption{
  Empirical Data Reconstruction Attack (DRA) of each protocol under the Spike distribution, evaluated over $n = 10{,}000$ users.
  As in the Zipf setting, MSS remains the most resistant to attack across all privacy levels.
  }
  \label{fig:asr_spike}
\end{figure*}

\begin{figure*}[!htb]
  \centering
  \begin{subfigure}{0.49\linewidth}
    \centering
    \includegraphics[width=\linewidth]{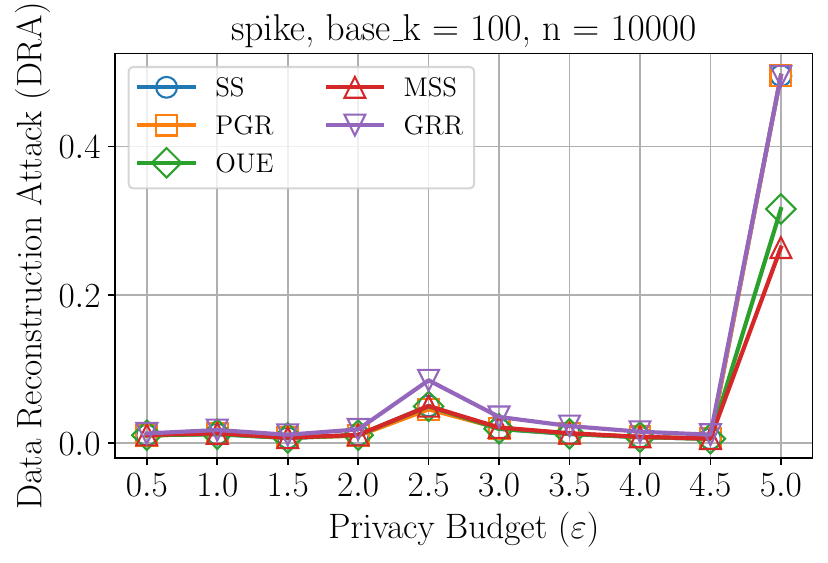}
    \caption{Spike, $\mathrm{base\_k} = 100$.}
    \label{fig:asr_spike_k_100}
  \end{subfigure}
  \hfill
  \begin{subfigure}{0.49\linewidth}
    \centering
    \includegraphics[width=\linewidth]{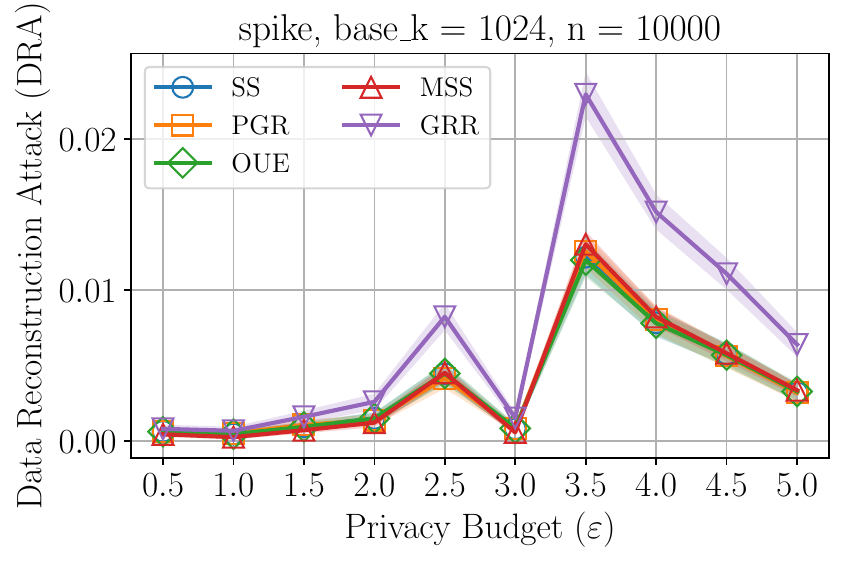}
    \caption{Spike, $\mathrm{base\_k} = 1{,}024$.}
    \label{fig:asr_spike_k_1024}
  \end{subfigure}
  \caption{
    Geometry-aligned ablation for PGR. 
    For each privacy level $\varepsilon$, the evaluation domain is set to the projective size $K(\varepsilon)$ induced by PGR, computed using $\mathrm{base\_k}=100$ and $1{,}024$. 
    All protocols are evaluated on this shared domain to remove the truncation mismatch that occurs when $k < K(\varepsilon)$.
    Under this aligned setting, PGR no longer exhibits the high-$\varepsilon$ spike observed in Figs.~\ref{fig:asr_zipf} and~\ref{fig:asr_spike}, while MSS remains the most resistant to reconstruction attacks.
    }
  \label{fig:asr_pgr_aligned}
\end{figure*}

\end{document}